\newif\ifdraft\draftfalse
\newif\iflater\latertrue  
\newif\ifpostreview\postreviewfalse
\newif\ifextended\extendedtrue

\documentclass[acmsmall]{acmart}
\usepackage{booktabs}   
\usepackage{xparse}
\usepackage{xspace}
\usepackage{xcolor}
\usepackage{relsize}
\usepackage{float}
\usepackage{cleveref}
\usepackage{syntax}
\usepackage{bm}
\usepackage{amsmath}
\usepackage{amssymb}
\usepackage{amsthm}
\usepackage{listings}
\usepackage{array}
\usepackage{mathpartir}
\usepackage{natbib}
\usepackage{graphicx}
\usepackage[inline]{enumitem}

\acmJournal{PACMPL}
\definecolor{dkblue}{rgb}{0,0.1,0.5}
\definecolor{dkgreen}{rgb}{0,0.5,0}
\definecolor{dkred}{rgb}{0.7,0,0}
\definecolor{dkpurple}{rgb}{0.7,0,0.4}
\definecolor{olive}{rgb}{0.4, 0.4, 0.0}
\definecolor{teal}{rgb}{0.0,0.5,0.5}
\definecolor{azure}{rgb}{0.0, 0.4, .8}
\definecolor{edoyellow}{rgb}{0.568, 0.568, 0.008}

\long\def\comment#1{}

\newcommand{\comm}[3]{\ifdraft\textcolor{#1}{[#2: #3]}\fi}

\newcommand{\bcp}[1]{\comm{dkpurple}{BCP}{#1}}

\newcommand{\Fuzzi}{Fuzzi\xspace}
\newcommand{\aprhl}{apRHL\xspace}
\newcommand{\tyint}[0]{\mathtt{int}}
\newcommand{\tyfloat}[0]{\mathtt{real}}
\newcommand{\tybool}[0]{\mathtt{bool}}

\newcommand{\laplace}[2]{\mathcal{L}_{#1}(#2)}

\newcommand{\cmdif}[3]{\bm{\mathtt{if}}\, #1\, \bm{\mathtt{then}}\, #2\, \bm{\mathtt{else}}\, #3\, \bm{\mathtt{end}}}
\newcommand{\cmdwhile}[2]{\bm{\mathtt{while}}\, #1\, \bm{\mathtt{do}}\, #2\, \bm{\mathtt{end}}}
\newcommand{\cmdskip}[0]{\bm{\mathtt{skip}}}

\newcommand{\cmdbmap}[6]{\bm{\mathtt{bmap}}(#1, #2, #3, #4, #5, #6)}
\newcommand{\cmdamap}[6]{\bm{\mathtt{vmap}}(#1, #2, #3, #4, #5, #6)}

\newcommand{\gsep}[0]{\,|\,}
\newcommand{\denote}[1]{{[\![#1]\!]}}
\newcommand{\hastype}[4]{#1 \vdash #2 \in_{#3} #4}
\newcommand{\tytriple}[3]{\left\{#1\right\}\,\allowbreak #2\,\allowbreak \left\{#3\right\}}
\newcommand{\aprhlstmt}[5]
{\vdash\, #1\,\allowbreak
              \sim_{#2}\,\allowbreak
          #3\,\allowbreak
              :\,\allowbreak #4
          \Rightarrow \allowbreak #5}
\newcommand{\ag}[1]{\langle #1 \rangle}

\newcommand{\sapprox}{\mathtt{approx}}
\newcommand{\pr}{\mathbb{P}}

\newcommand{\fv}{\mathtt{fvs}}
\newcommand{\mv}{\mathtt{mvs}}
\newcommand{\distr}{\mathop{\bigcirc}}
\newcommand{\term}[2]{#1\vdash #2\, \mathtt{term}}
\newcommand{\determ}{\mathtt{determ}}
\newcommand{\linear}{\mathtt{linear}}
\newcommand{\shape}{\mathtt{shape}}

\newcommand{\llet}{\mathtt{let}}
\newcommand{\lin}{\mathtt{in}}
\newcommand{\literal}{\mathtt{literal}}
\newcommand{\widegsep}{\,\;\gsep\;\,}

\newenvironment{mathparsmall}
  {\begin{small}\begin{mathpar}}
  {\end{mathpar}\end{small}}

\definecolor{gray}{rgb}{0.4, 0.4, .4}
\newcommand{\tinytext}[1]{\mbox{\tiny\color{gray}{#1:}}}

\newtheorem{lem}{Lemma}

\newtheorem{defn}{Definition}

\newcolumntype{L}{>{$}l<{$}}
\newcolumntype{C}{>{$}c<{$}}
\newcolumntype{R}{>{$}r<{$}}

\author{Hengchu Zhang}
\affiliation{
  \institution{University of Pennsylvania}            
  \country{USA}
}
\email{hengchu@seas.upenn.edu}          

\author{Edo Roth}
\affiliation{
  \institution{University of Pennsylvania}            
  \country{USA}
}
\email{edoroth@seas.upenn.edu}          

\author{Andreas Haeberlen}
\affiliation{
  \institution{University of Pennsylvania}            
  \country{USA}
}
\email{ahae@cis.upenn.edu}          

\author{Benjamin C. Pierce}
\affiliation{
  \institution{University of Pennsylvania}            
  \country{USA}
}
\email{bcpierce@cis.upenn.edu}          

\author{Aaron Roth}
\affiliation{
  \institution{University of Pennsylvania}            
  \country{USA}
}
\email{aaroth@cis.upenn.edu}          

\begin{document}

\title{Fuzzi: A Three-Level Logic for Differential Privacy}

\keywords{Differential privacy, typechecking, static analysis, \aprhl, Fuzz,
\Fuzzi}

\begin{abstract}

Curators of sensitive datasets sometimes need to know whether queries against
the data are {\em differentially private}~\cite{Dwork:2006:CNS:2180286.2180305}.
Two sorts of logics have been proposed for checking this property: (1) {\em type
systems} and other static analyses, which fully automate straightforward
reasoning with concepts like ``program sensitivity'' and ``privacy loss,'' and
(2) full-blown program logics such as \aprhl (an approximate, probabilistic,
relational Hoare logic)~\cite{Barthe:2016:PDP:2933575.2934554}, which support
more flexible reasoning about subtle privacy-preserving algorithmic techniques
but offer only minimal automation.

We propose a {\em three-level logic} for differential privacy in an
imperative setting and present a prototype implementation
called \Fuzzi.  \Fuzzi's lowest level is a general-purpose logic; its middle
level is \aprhl; and its top level is a novel {\em sensitivity logic} adapted
from the linear-logic-inspired type system of Fuzz, a differentially private
functional language~\cite{Reed:2010:DMT:1932681.1863568}. The key novelty is a
high degree of integration between the sensitivity logic and the two lower-level
logics: the judgments and proofs of the sensitivity logic can be easily translated
into \aprhl; conversely, privacy properties of key algorithmic
building blocks can be proved manually in \aprhl and the base logic, then
packaged up as typing rules that can be applied by a checker for the sensitivity
logic to automatically construct privacy proofs for composite programs of
arbitrary size.

We demonstrate \Fuzzi's utility by implementing four different private
machine-learning algorithms and showing that \Fuzzi's checker is able to
derive tight sensitivity bounds.


\end{abstract}

\maketitle

\section{Introduction}
\label{sec:intro}

Differential privacy~\cite{Dwork:2006:CNS:2180286.2180305} has become the gold
standard for privacy-preserving statistical analysis in the academic
community, and it is being adopted by a growing number of industry and
government organizations, including Apple~\cite{AppleDP},
Google~\cite{Erlingsson:2014:RRA:2660267.2660348}, Microsoft~\cite{MSDP} and the
US Census Bureau~\cite{CensusDP}. Differential privacy makes minimal assumptions
about an adversary's knowledge, allowing analysts to quantitatively estimate
privacy loss.  However, the reasoning needed to correctly achieve differential
privacy can be rather subtle, as multiple errors in published algorithms
attest~\cite{DBLP:journals/corr/ChenM15e, DBLP:journals/corr/LyuSL16}.

\citet{Barthe:2016:PDP:2933575.2934554} developed the first program logic for formalizing proofs of differential privacy for imperative
programs, called \aprhl (Approximate, Probabilistic, Relational Hoare Logic).  The abstractions provided by \aprhl are expressive enough to
capture the essence of many complex differentially private algorithms,
allowing experts to prove differential privacy for small, tricky code
sequences at a fairly high level of abstraction. However, proving
differential privacy in \aprhl for large programs can be a rather tedious
endeavor. Fortunately, in many proofs for larger private data analysis
programs (either in \aprhl or on paper), the expert knowledge of
differential privacy is concentrated in the analysis of small differentially
private subroutines, while the rest of the proof basically just propagates
``sensitivity'' information and aggregate privacy costs between subroutines.
This suggests that one could considerably increase the range of possible
use cases, especially for analysts who are not privacy experts, by combining
a small but extensible set of building blocks (and the corresponding manual
proofs) with a largely automated analysis that mechanically completes the proof
for a given program.

To enable this approach, we build a new layer of abstraction
over \aprhl to automate the mechanical parts of this process. This layer
tracks \textit{sensitivities} for program
variables and \textit{privacy costs} of commands using Hoare-triple-style
proof rules. This information about sensitivity and privacy cost has a
direct translation to lower-level \aprhl assertions. This allows
information in the higher-level logic to seamlessly interact with expert
proofs of differential privacy that have been carried out using the two lower layers.
Since the top layer is
entirely automated, we will often refer to it as a type system (and to its
proof rules as typing rules).

We use the term \emph{mechanisms} to refer to building blocks of
differentially private
programs. Many differentially private mechanisms can be viewed as
parameterized program templates, where the differential privacy properties
depend on properties of the instantiation parameters, which can themselves be
program expressions or commands. In order to integrate expert reasoning about
such mechanisms, we develop a framework for expressing program templates and the
corresponding parameterized proofs of differential privacy. This allows experts
to extend the sensitivity type system with a specialized typing rule for each
template, allowing non-expert programmers to write application programs that
combine these templates in straightforward ways.  This framework uses \aprhl
directly to give structured proofs of privacy property, while using the
general-purpose base logic to establish lower-level semantic properties that go
beyond the capabilities of \aprhl.

We instantiate these ideas in the design and implementation of \Fuzzi, a
small imperative language for differentially private queries with
automatic and extensible typechecking.  Following a brief review of
technical background on differential privacy (\Cref{sec:review}) and a
high-level overview of \Fuzzi's design (\Cref{sec:overview}), we offer the
following contributions:
\begin{enumerate}
\item We propose a high-level {\em sensitivity logic} for tracking
differential privacy (\Cref{sec:senslogic}). This logic is expressive enough
to capture detailed
sensitivity properties for a simple imperative core language; its soundness
is established via a straightforward embedding into \aprhl.
\item We show how to connect manual proofs for privacy properties of
algorithmic building blocks to the sensitivity logic and develop proofs for
 several mechanisms that transform private datasets, plus a mechanism that
 aggregates
 privacy costs better than straightforward composition (\Cref{sec:extension}).

\item Using a prototype implementation of \Fuzzi (\Cref{sec:impl}), we
implement private machine learning algorithms from four different classes of
learning methods (discriminative models, ensemble models, generative models
and instance-based learning) and show that \Fuzzi's checker is able to
derive tight sensitivity bounds (\Cref{sec:eval}).
\end{enumerate}
\Cref{sec:limitations} discusses limitations of the current design.
\Cref{sec:related-works,sec:future} survey related and future work.
%


\section{Background}
\label{sec:review}

\subsection{Differential Privacy}

Differential privacy is an indistinguishability property of randomized programs
on neighboring input datasets. Informally, a function is differentially private
if removing or adding a single row in the input results in at most a small
change in the output distribution.
\begin{defn}[Neighboring Dataset]
Two datasets are neighbors if one can be transformed into the other by adding or removing
a single row of data.
\end{defn}\vspace{-0.4em}
Let $D$ and $D'$ be two neighboring datasets, and let $f$ be a randomized
program. The output of $f$ is a sample from some distribution parameterized by
the input datasets. We write $f(D)$ and $f(D')$ for these two distributions.\vspace{-0.4em}
\begin{defn}[$(\epsilon, \delta)$-Differential Privacy~\cite{Dwork:2006:CNS:2180286.2180305}]
\label{defn:dp}
The program $f$ is $(\epsilon, \delta)$-differentially private if, for any
set of possible outputs $E$, the probability of observing $E$ satisfies the relation
$\pr_{x \sim f(D)}[x \in E] \leq e^\epsilon \pr_{x \sim f(D')}[x \in E] + \delta$.
\end{defn}\vspace{-0.4em}
The parameters $\epsilon$ and $\delta$ quantify different aspects of the privacy
cost of a differentially private computation. Informally, the value of
$\epsilon$ measures the ability of an observer to distinguish whether $f$ was
run with $D$ or $D'$ after observing $E$ in the ``common case'', while $\delta$
serves as an upper bound on the probability that $f$ fails to provide the
privacy guarantee implied by $\epsilon$. The parameter $\epsilon$ is typically
taken to be a small constant (say, $1$), whereas $\delta$ must be set so that
$\delta \ll 1/n$, where $n$ is the number of dataset rows, in order for the
privacy guarantees to be non-trivial (otherwise an algorithm which outputs a
dataset row uniformly at random satisfies $(0,\delta)$-differential privacy).

\subsection{Sensitivity}
The notion of \textit{sensitivity} is crucial to differential privacy. Many
differentially private mechanisms release data by adding noise proportional to
the sensitivity of a private value. In \Fuzzi, the term ``sensitivity''
specifically refers to an upper bound on the \textit{distance} between the
values held by some variable between any two runs.

Distance may be calculated differently for values of different types. For
primitive values with type $\tyint$ and $\tyfloat$, distance is the magnitude of
the two values' difference. However, for arrays, there are two important
distance definitions for differential privacy: {\em database distance} and {\em
L1 distance}. The database distance measures the number of rows that need to be
added or removed in order to make two datasets indistinguishable up to
permutation; while the L1 distance measures the sum of element-wise distance
between vectors. To avoid confusion in later discussions, we refer to
arrays for which distance is intended to be measured as database distance
as \textit{bags} and arrays with L1 distance as \textit{vectors}. When we come
to defining the type system, we will write $\{\tau\}$ for the type of bags
holding values of type $\tau$ and $[\tau]$ for vectors of $\tau$
(\Cref{fig:syntax-core}). As an example, the two arrays $[1, 2, 5]$ and $[1, 3, 4]$
have vector distance $2$, but they have bag distance $4$, since we need to
remove elements $2$ and $5$ and add elements $3$ and $4$ to the first bag in
order to make it a permutation of the second one.%
\footnote{%
Database distance can actually be viewed as just L1 distance on a
different representation of datasets. The differential privacy literature
sometimes uses the ``histogram representation'' for datasets. For a universe of
possible elements $U$, the histogram representation maps each $x \in U$ to a
count of how many times $x$ appears, and the L1 distance of this representation
corresponds to the database distance. However, in order to keep \Fuzzi's
semantics minimal as a core language, we choose to represent datasets as arrays,
rather than maps from records to counts.}

Formally, we write $d_\tau$ for the
distance function at type $\tau$, with type $\tau \times \tau \rightarrow \mathbb{R}^+ \cup
\{\infty\}$---i.e., it maps
two values of type $\tau$ to a non-negative real number or infinity.\vspace{-0.4em}
\begin{defn}[Vector distance]
\label{defn:array-dist}
  If $a_1$ and $a_2$ are two vectors of the same length and their elements have
  type $\tau$, then the vector distance is defined as $d_{[\tau]}(a_{1}, a_{2})
  = \sum_{i = 0}^{L - 1} d_{\tau}(a_{1}[i], a_{2}[i])$ where $L$ is the
  length of both vectors.
Vectors of different lengths are assigned the distance $\infty$.
\end{defn}\vspace{-0.4em}
\begin{defn}[Bag distance]
\label{defn:bag-dist}
  Let $a_{1}$ and $a_{2}$ be two bags; their distance is defined as
  $d_{\{\tau\}}(a_{1}, a_{2}) = |a_{1} \setminus a_{2}| + |a_{2} \setminus
  a_{1}|$.
\end{defn}\vspace{-0.4em}
The backslash operator is multiset difference. Note that bag distance
(unlike vector distance) is meaningful for bags of different sizes. This is
the same database distance introduced earlier.


\subsection{Laplace Mechanism}
The Laplace mechanism is an essential tool for releasing private data with
bounded sensitivities~\cite{Dwork:2006:CNS:2180286.2180305,
  Dwork:2014:AFD:2693052.2693053}.  \Fuzzi provides access to the Laplace
mechanism through the {\em sampling assignment} command
$x = \laplace{b}{e}$, which adds noise to the value of $e$ and assigns that
value to the variable $x$, with the constant literal $b$ determining the
scale of the noise. Adding noise scaled with $b$ to a value with sensitivity
$s$ incurs a privacy cost of $(s/b, 0)$.  \Fuzzi's type system will
statically keep track of each usage of the Laplace mechanism and report an
upper bound of the total privacy cost as part of a program's type.


\section{Overview}
\label{sec:overview}
\subsection{Core Language}
\label{sec:fuzzi-core-language}
\begin{figure}[t]
\centering
\begin{minipage}{0.45\textwidth}
\centering
\begin{align*}
  \sigma \;\;:=\;\; &\tyint \widegsep \tyfloat \widegsep \tybool\\
  \tau   \;\;:=\;\; &\sigma \widegsep [\tau] \widegsep \{\tau\} \\
  \mathtt{op} \;\;:=\;\; & + \widegsep - \widegsep \cdot \widegsep /
  \widegsep \mathtt{\&\!\&} \widegsep\, \mathtt{||} \widegsep \\
   &< \widegsep \leq
  \widegsep > \widegsep \geq \widegsep \neg\\
    e, i   \;\;:=\;\; &x \widegsep \mathtt{lit} \widegsep e\, \mathtt{op}\, e
  \widegsep e[i] \widegsep e.length
\end{align*}
\end{minipage}
\begin{minipage}{0.45\textwidth}
  \centering
  \begin{align*}
  c \;\;:=\;\; &x = e
    \widegsep x[i] = e
    \widegsep
    x.length = e
    \widegsep \\ & x = \laplace{b}{e} \widegsep\\
     &\cmdif{e}{c}{c}
    \widegsep \\
    &\cmdwhile{e}{c} \widegsep
    c;c
  \end{align*}
\end{minipage}
\caption{Core Language Syntax}\vspace{-1.5em}
\label{fig:syntax-core}
\end{figure}
The core of \Fuzzi is a simple imperative programming language with while loops,
conditionals, and assignments (\Cref{fig:syntax-core}). It has just a few
built-in data types: reals, integers, booleans, and arrays (whose elements
can be reals, integers, booleans, or nested arrays).  Programs can modify the
length of arrays through
assignments of the form $x.\mathit{length} = e$. When the value of $e$ is less
than the current length of $x$, the array is truncated; and when the value of
$e$ is greater than the length of $x$, the array is padded with default values
that are of the same data type as elements in $x$. If $e$ evaluates to a
negative number, the length assignment diverges.

One slightly unusual feature of \Fuzzi is that {\em all} assignments are copying
assignments, including assignments to array variables. For
example, if $x$ holds an array value, then the assignment $y = x$ sets $y$ to a
copy of $x$, instead of making both $x$ and $y$ point to the same underlying
array.
We make this choice to avoid reasoning about sharing, which we
consider as out of scope for this work.

The special command $x = \laplace{b}{e}$ performs probabilistic
assignment to $x$ by sampling from a Laplace distribution centered at the value
of $e$, with width equal to the value of $b$ (which must be a
real-valued literal\ifpostreview\bcp{We might want to comment on how
  serious a restriction this is.}\fi).

\subsection{Type System}
From now on, we refer to the sensitivity logic as a {\em type system} to
emphasize
that it is a specialized and automated layer that tracks sensitivity and
privacy cost as types.
\ifpostreview\bcp{I'd prefer to call it a type system all the way
  through (starting in the abstract).  Also, I'm a little nervous about
  calling it a ``Sensitivity Type
System'' (or Logic), when it is critically tracking {\em both} sensitivity and
privacy loss.  Not sure what would be a better name, though.}\fi
The main data structure manipulated by \Fuzzi's type
system is {\em typing contexts}---maps from program variables to
sensitivities, represented as non-negative reals extended with infinity.
$$
  \phi   \;:=\; r \in \mathbb{R}^{\geq 0} \cup \{\infty\}\qquad\qquad
  \tau   \;:=\; \sigma \widegsep [\tau] \widegsep \{\tau\} \qquad\qquad
  \Gamma \;:=\;  \varnothing \widegsep x:_\phi \tau, \Gamma
$$
A typing context $\Gamma$ should be interpreted as a {\em relation} between
two different run-time stores---intuitively, the stores
arising from
two ``neighboring'' executions starting with neighboring
initial states.  For each variable $x$ with sensitivity $\Gamma(x)$, the
values in the two stores must be no more than $\Gamma(x)$ apart.

Typing judgements for commands have the form
$\tytriple{\Gamma}{c}{\Gamma', (\epsilon, \delta)}$, meaning that, if the
distance of the values in two run-time stores are described by the
initial typing context $\Gamma$, then executing $c$ will either both diverge or
else both terminate with two final stores described by $\Gamma'$, along the way
incurring a privacy cost of $(\epsilon, \delta)$.

For example, the typing rule for commands of the
form $x = e$ computes the sensitivity of $e$ using sensitivities of its free
variables, and maps $x$ to this sensitivity of $e$ in the output context.
\begin{mathparsmall}
\inferrule[Assign]
            {\hastype{\Gamma}{e}{s}{\tau}}
            {\tytriple{\Gamma}{x = e}{\Gamma[x \mapsto s], (0, 0)}}
\end{mathparsmall}%
The typechecker also computes the privacy cost
$(\epsilon, \delta)$ incurred by the analyzed command. In the case of
assignment, no privacy cost is incurred, so the output from the typechecker
after processing $x = e$ is the updated typing context $\Gamma[x\mapsto s]$ and
the pair of privacy costs $(0, 0)$, where $s$ is the derived sensitivity of $e$
under $\Gamma$.

A more interesting typing rule is the one for sequence commands of the form
$c_1; c_2$.  This rule chains together the typing judgements for each of the
commands, using the output context $\Gamma_i$ from analyzing $c_i$ as the input
context for processing the next command $c_{i+1}$. The privacy cost incurred by
the whole program is the sum of privacy costs $(\epsilon_i, \delta_i)$ incurred
by each $c_i$, following the ``simple composition theorem'' for differential
privacy~\cite{Dwork:2006:CNS:2180286.2180305}.
\begin{mathparsmall}
  \inferrule[Sequence]
            {\tytriple{\Gamma_1}{c_1}{\Gamma_2, (\epsilon_1, \delta_1)} \quad
             \tytriple{\Gamma_2}{c_2}{\Gamma_3, (\epsilon_2, \delta_2)}
            }
            {\tytriple{\Gamma_1}{c_1; c_2}{\Gamma_3, (\epsilon_1 +
                \epsilon_2, \delta_1 + \delta_2)}}
\end{mathparsmall}%
There are also core typing rules for simple loops and conditionals that do
not branch on sensitive data; we will see these in \Cref{sec:typing-fuzzi}.

\subsection{Typing Differentially Private Mechanisms}
The privacy properties of interesting differentially private mechanisms are
generally too subtle to be tracked by the core type system.  In \Fuzzi, such
mechanisms can be defined as extensions and equipped with specialized typing
rules whose soundness is proved manually.  Such proofs typically involve
reasoning about relational properties for distributions, as well as
aggregating privacy costs. The program logic \aprhl is tailored to tackle
both problems, making it a good choice for rigorous manual proofs of
differential privacy.

An \aprhl judgement has the form
$\aprhlstmt{c_1}{(\epsilon, \delta)}{c_2}{\Phi}{\Psi}$, where $c_1$ and $c_2$
are two commands to be related, $\Phi$ and $\Psi$ are relational assertions that
state pre- and post-conditions relating the program states before and after
executing $c_1$ and $c_2$. A sound \aprhl judgement can be roughly interpreted
as: if (1) some pair of program states satisfy the pre-condition $\Phi$, and (2)
executing $c_1$ in the first state terminates iff executing $c_2$ in the second
state does, then the pair of states after executing $c_1$ and $c_2$ will satisfy
the post-condition $\Psi$, incurring privacy cost $(\epsilon, \delta)$. To
express differential privacy as a post-condition, we can simply state
$\mathit{out}\ag{1} = \mathit{out}\ag{2}$ for the output of the programs.

Given a typing context $\Gamma$, we can interpret $\Gamma$ as a relation on
program states.  Writing $\ag{1}$ and $\ag{2}$ after variables to refer to
their values in the first or second execution, we translate
$x :_\sigma \tau \in \Gamma$ to the assertion
$d_\tau(x\ag{1}, x\ag{2}) \leq \sigma$; the conjunction of all these
pointwise distance assertions forms an \aprhl assertion that corresponds to
$\Gamma$.

Conversely, to connect manual proofs in \aprhl with the type system, we
phrase their premises and conclusions as typing judgements.  Indeed, we use
\aprhl not only for extensions but also for the soundness proofs of the core
typing rules. As a result, the privacy proofs implicitly constructed by the
typechecker are combinations of \aprhl proof objects, some of them generated by
the typechecker, others written manually by experts.

\subsection{Example}
To give a first taste of \Fuzzi's differential privacy typechecking process,
we present a simple program that computes a private
approximation for the average
income of a
group through private estimations of the
group's size and sum. First, we estimate the group's size with the Laplace mechanism.
\begin{lstlisting}[xleftmargin=1em]
size = &$\laplace{1.0}{\mathtt{group.length}}$&;
\end{lstlisting}
Assuming
that \lstinline|group| is a dataset with sensitivity $1$, \Fuzzi's
typechecker deduces its size is $1$-sensitive.  Applying
the Laplace mechanism then incurs a $(1.0, 0)$-privacy cost.

Next, we sum the group's incomes using the mechanism \lstinline|bsum|, pronounced ``bag sum'', which clips each
income value so that its magnitude is at most a given constant (here 1000).
\begin{lstlisting}[xleftmargin=1em]
bsum(group, sum, i, temp, 1000);
\end{lstlisting}
This clipping step ensures the sum does not vary too much on
neighboring datasets. Without clipping, a single outlier could sway the sum
substantially, revealing the outlier's existence, and violating
differential privacy. The parameters \lstinline{sum}, \lstinline{i}, and
\lstinline{temp} specify the names of variables that \lstinline{bsum} can
use for internal purposes. It is the programmer's responsibility to
make sure they do not clash with variables used elsewhere in the
program.  (It should not be hard to fix this infelicity by
making extensions themselves deal with fresh variable generation, but doing
so will introduce a few additional technicalities so we leave it for future work.)
\ifpostreview
\bcp{This is a wart that we should really clean up: It should be
  perfectly easy to make the extension
  mechanism generate its own variable names!}
\fi

The command \lstinline|bsum(...)| refers to an extension that expands to a
sequence of plain core-language commands implementing summing up a bag of
numbers with clipping:
\begin{lstlisting}
extension bsum(in, out, idx, t_in, bound) {
  idx = 0;
  out = 0.0;
  while idx < in.length do
    t_in = in[idx];
    if t_in < -1.0 * bound then
      out = out - bound;
    else
      if t_in > bound then
        out = out + bound;
      else
        out = out + t_in;
      end
    end;
    idx = idx + 1;
  end
};
\end{lstlisting}
This specifies the name of the extension, the names of its
parameters (which range over \Fuzzi expressions and commands), and how
it expands to core \Fuzzi
commands.  During typechecking, extension applications are replaced by their
expansions, with extension variables substituted by the snippets of \Fuzzi
syntax provided as parameters.
\ifpostreview
(The expansion process also leaves hints for the
typechecker around the expanded \Fuzzi core commands. The hints contain the
name of the extension,
and the list of parameters that the extension was invoked
with.\bcp{This is a bit mysterious.}\fi)

The typing rule for \lstinline{bsum} is:
\begin{mathparsmall}
\inferrule{\literal\,\mathit{bound} \quad \mathit{bound} \geq 0\\\\
           \phi = \Gamma(\mathit{in}) \quad \Gamma_\mathit{out}
           = \Gamma[\mathit{out} \mapsto \phi\cdot \mathit{bound}][i, \mathit{t_{in}} \mapsto \infty]
          }
          {\tytriple
            {\Gamma}
            {\bm{\mathtt{bsum}} (\mathit{in}, \mathit{out}, i,
              t_\mathit{in}, \mathit{bound})}
            {\Gamma_\mathit{out}, (0, 0)}
          }
\end{mathparsmall}%
It requires the last parameter \lstinline|bound| to be a
non-negative literal value---non-negative because \lstinline|bound| specifies
the clipping magnitude, and literal because the sensitivity of the output
variable depends on \lstinline|bound|. The inference rule updates the
sensitivity of the output sum variable to the product of \lstinline|bound| and
the sensitivity of $\Gamma(in)$. Intuitively, since up to $\phi$ elements may be
added or removed from $\mathit{in}$, and each can contribute up a value with
magnitude up to $\mathit{bound}$ toward the sum, the sum value itself will vary
by at most $\phi\cdot\mathit{bound}$. This intuition can be made rigorous,
as we show in
\ifextended
\Cref{ap:bag-sum-proof}.
\else
Appendix D.4 of the extended version.
\fi

The Haskell implementation of the \Fuzzi typechecker is likewise extended
with a piece of code implementing the typing rule as a function that
transforms an input typing context to an output typing context and privacy
costs.

\ifpostreview\bcp{I think this paragraph can be deleted.}
In general, the \Fuzzi syntax allows extension invocation in addition to core
commands:
\begin{align*}
c \;\;:=\;\; & \dots \widegsep \mathit{ext}(\mathit{param}_1, \dots, \mathit{param}_n) \\
\mathit{param} \;\;:=\;\; &e \widegsep c
\end{align*}
\fi

Continuing the example, we next compute differentially private estimates of the
clipped sums and calculate the group's average income using the size and sum
estimates:
\begin{lstlisting}[xleftmargin=1em]
noised_sum = &$\laplace{1000.0}{\mathtt{sum}}$&;
avg = noised_sum / size;
\end{lstlisting}
The \lstinline|sum| variable is $1000$-sensitive, so
releasing \lstinline|noised_sum| incurs another $(1.0, 0)$-privacy cost. The
typechecker
reports an aggregate privacy cost of $(2.0, 0)$.

\section{Sensitivity Type System}
\label{sec:senslogic}

\subsection{Notation and definitions}
\label{sec:def}
Throughout the paper, we will use the operator $\denote{\cdot}$ to denote the
semantic function for commands and expressions in \Fuzzi. We use the notation
$\distr S$ to denote sub-distributions over values in $S$. We will use
the letter
$M, N$ to stand for program states, which are finite maps from variable names to
the values they hold, and use the letter $\mathbb{M}$ to stand for the set of
all program states.

The semantics of a \Fuzzi program $c$ is a function from program states to
sub-distributions over program states $\denote{c} : \mathbb{M} \rightarrow
\distr \mathbb{M}$. Each type in \Fuzzi is associated with a set of values: $\tyint$ with the set
$\mathbb{Z}$, $\tyfloat$ with the set $\mathbb{R}$, and $[\tau]$ and $\{\tau\}$
with the set of finite sequences of values associated with $\tau$. The meaning
of a \Fuzzi expression $e$ with type $\tau$ is a partial function from program
states to associated values of that type $\denote{e}
: \mathbb{M} \rightharpoonup \tau$. Partiality of expressions stems from invalid
operations such as arithmetic between incompatible values, and out-of-bound
indexing. The complete definition of \Fuzzi semantics can be found
in
\ifextended
\Cref{ap:semantics}.
\else
Appendix A of the extended version.
\fi
Recall from Section~\ref{sec:fuzzi-core-language} that \Fuzzi assignments are
copy-assignments for all values, including vectors and bags.

Fuzzi's semantics directly follows from the work
of \citet{Barthe:2016:PDP:2933575.2934554}. It is worth noting that the
original \aprhl developed by \citet{Barthe:2016:PDP:2933575.2934554} only
reasons with discretized Laplace distributions, and \Fuzzi shares this
restriction in its semantic model. A later model based on category theory
enhances \aprhl's proof rules for continuous
distributions~\cite{SATO2016277}. However, the underlying proof method of this
model is not compatible with Fuzzi's development, and only recently
have new abstractions been proposed to generalize the original \aprhl proof
methods to continuous distributions~\cite{DBLP:journals/corr/abs-1710-09010}.

\label{sec:typing-fuzzi}
\Fuzzi's typing context $\Gamma$ tracks both the data type and the sensitivity of
variables. Typechecking involves checking data types as well as computing
sensitivities. We refer to data typechecking as \textit{shape checking} and
sensitivity computations as \textit{sensitivity checking}. We will elide details
of shape checking since it is the standard typechecking that rules out operations
between values of incompatible types. To emphasize \textit{sensitivity checking}
in \Fuzzi, and to reduce clutter in syntax, we will write $\Gamma(x)$ for the
sensitivity of the variable $x$ under the typing context $\Gamma$, and we write
$\Gamma[x \mapsto s]$ for a typing context which updates variable $x$'s
sensitivity to $s$, but does not alter its data type. We overload this syntax
when we update a set of variables $xs$ to the same sensitivity
$\Gamma[xs \mapsto s]$. We also overload the notation $\Gamma(e)$ to denote the
derived sensitivity of expression $e$ under typing context $\Gamma$. When we
need to refer to the data type of expression $e$, we will use the full typing
judgment of an expression $\hastype{\Gamma}{e}{\phi}{\tau}$, which we pronounce
``expression $e$ has sensitivity $\phi$ and type $\tau$ under context
$\Gamma$.''

We use the notation $\shape(\Gamma)$ to extract the shape checking context from
a typing context $\Gamma$, dropping all sensitivity annotations.

\subsection{Typing Expressions}
\label{subsec:typing-expressions}

In order to compute sensitivity updates throughout sequences of commands, the
type system needs to first compute sensitivities for expressions used within
each command.
We discuss the typing rules for addition and multiplication here as
examples. Intuitively, if the values of two expressions $e_l$ and $e_r$ can
each vary by $1$, then their sum can vary by at most $2$ (the sum of their
individual sensitivities) by the triangle inequality; and if the value of
$e$ can vary by at most $1$, then multiplying $e$ by a literal constant $k$
results in a value that can vary by at most $k$. The rules \textsc{Plus},
and \textsc{Mult-L-Constant} in~\Cref{fig:arith} capture these cases.

\begin{figure}[t]
  \begin{mathparsmall}
  \inferrule[Plus]
            {\hastype{\Gamma}{e_l}{s}{\tau} \quad
              \hastype{\Gamma}{e_r}{t}{\tau} \\\\
            \tau = \tyint \lor \tau = \tyfloat}
            {\hastype{\Gamma}{e_l+e_r}{s + t}{\tau}}

  \inferrule[Mult-L-Constant]
            {\literal\,k\quad
              \hastype{\Gamma}{e_r}{t}{\tau} \\\\
              \tau = \tyint \lor \tau = \tyfloat}
            {\hastype{\Gamma}{k \cdot e_r}{k \cdot t}{\tau}}


  \inferrule[Mult]
            {\hastype{\Gamma}{e_l}{s}{\tau} \quad
              \hastype{\Gamma}{e_r}{t}{\tau} \\\\
              \tau = \tyint \lor \tau = \tyfloat
            }
            {\hastype{\Gamma}{e_l \cdot e_r}{\sapprox(s, t)}{\tau}}
  \end{mathparsmall}
  \caption{Arithmetic Expression Typing Rules}
  \label{fig:arith}
\end{figure}

There are also expressions for which we cannot give precise sensitivity
bounds. For instance, if one of the operands for a multiplication between
$e_l$ and $e_r$ is
sensitive, then, without knowing the exact value of the other operand,
we cannot \textit{a priori} know how much the value of entire product can
change. This case is captured by the \textsc{Mult} rule, where the function
$\sapprox$ is defined by the equations
\begin{center}
\begin{tabular}{RLCL}
  \sapprox & (0,\hspace{0.5em} 0) & = & 0\\
  \sapprox & (s_1, s_2)\quad \mathrm{if}\, s_1 + s_2 > 0 & = & \infty
\end{tabular}
\end{center}
which conservatively take the sensitivity to be $\infty$ if at least one side of
the expression is sensitive.

\Fuzzi provides bag and vector index
operations, and \Fuzzi's typechecker supports sensitivity checking for
lookup
expressions on bags and vectors. These typing rules use the definition of bag
and vector distances to establish sound upper bounds of sensitivities on lookup
expressions.
\begin{mathparsmall}
\inferrule[Vector-Index]
          {\hastype{\Gamma}{e}{\phi}{[\tau]} \quad \hastype{\Gamma}{i}{0}{\tyint} \\\\ \phi < \infty}
          {\hastype{\Gamma}{e[i]}{\phi}{\tau}}

\inferrule[Bag-Index]
          {\hastype{\Gamma}{e}{0}{\{\tau\}} \quad \hastype{\Gamma}{i}{0}{\tyint}}
          {\hastype{\Gamma}{e[i]}{\infty}{\tau}}
\end{mathparsmall}%
\noindent The \textsc{Vector-Index} rule applies when the lookup expression is
$0$-sensitive. A $0$-sensitive index value must be the same across two
 executions, and the distance between two values at the same position must be
 bounded by the overall sensitivity of the vector itself according
 to \Cref{defn:array-dist}. As an example, given two vectors $[1, 2, 3]$ and
 $[1, 2, 4]$, if we indexed both vectors at the last position, then the
 resulting values $3$ and $4$ are at distance $1$ apart, which is bounded by the
 distance between the original vectors. The premise $\phi < \infty$ is necessary
 to ensure the indexed arrays have the same length in both executions, so that
 the lookup expression terminates in one execution if and only if it terminates
 in the other. We refer to this property as \textit{co-termination}.  It is
 discussed in \Cref{sec:fuzzi-aprhl}.

It may be surprising that \Fuzzi's typechecker only accepts bag lookup
operations
over non-sensitive bags. This is due to requirement of
\textit{co-termination} and the fact that bags with non-zero sensitivities
may have different lengths in neighboring runs. To see why the
bag lookup expression has sensitivity $\infty$, consider two bags $[1, 100, 2]$
and $[1, 2, 100]$; these are at distance 0, but if we access both bags with
index 1, the resulting values $100$ and $2$ are distance $98$ apart.

\subsection{Typing Commands}

The typing judgments for commands has the form $\tytriple{\Gamma}{c}{\Gamma',
(\epsilon, \delta)}$. We can think of these judgments as a
Hoare-triple---$\Gamma$ is a pre-condition of the program $c$, and $\Gamma'$ is
a post-condition for $c$---annotated with $(\epsilon, \delta)$, the total
privacy cost of running $c$.

\begin{figure}[t]
\begin{mathparsmall}
  \inferrule[Assign]
            {\hastype{\Gamma}{e}{\phi'}{\tau}}
            {\tytriple{\Gamma}{x = e}{\Gamma[x \mapsto \phi'], (0, 0)}}

  \inferrule[Assign-Vector-Index]
            {\hastype{\Gamma}{x}{\phi}{[\tau]}
             \quad \hastype{\Gamma}{e}{\sigma}{\tau} \\\\
             \phi < \infty \quad
             \hastype{\Gamma}{i}{0}{\tyint}}
            {\tytriple{\Gamma}{x[i] = e}{\Gamma[x \mapsto \phi + \sigma], (0, 0)}}


  \inferrule[Assign-Vector-Length]
            {\hastype{\Gamma}{x}{\phi}{[\tau]}
             \quad \hastype{\Gamma}{e}{0}{\tyint}}
            {\tytriple{\Gamma}{x.length = e}{\Gamma, (0, 0)}}

  \inferrule[Assign-Bag-Length]
            {\hastype{\Gamma}{x}{\phi}{\{\tau\}}
             \quad \hastype{\Gamma}{e}{0}{\tyint}}
            {\tytriple{\Gamma}{x.length = e}{\Gamma[x\mapsto \infty],(0, 0)}}

  \inferrule[Laplace]
            {\hastype{\Gamma, x:_\phi \tyfloat}{e}{\phi'}{\tyfloat}}
            {\tytriple{\Gamma, x:_\phi \tyfloat}{x = \laplace{b}{e}}
             {\Gamma, x:_0 \tyfloat, (\phi'/b, 0)}}

  \inferrule[Skip]
            { }
            {\tytriple{\Gamma}{\cmdskip}{\Gamma, (0, 0)}}

  \inferrule[Sequence]
            {\tytriple{\Gamma_1}{c_1}{\Gamma_2, (\epsilon_1, \delta_1)} \quad
             \tytriple{\Gamma_2}{c_2}{\Gamma_3, (\epsilon_2, \delta_2)}
            }
            {\tytriple{\Gamma_1}{c_1; c_2}{\Gamma_3, (\epsilon_1 + \epsilon_2, \delta_1 + \delta_2)}}

  \inferrule[If]
            {\tytriple{\Gamma}{c_t}{\Gamma_t, (\epsilon_t, \delta_t)} \quad
             \tytriple{\Gamma}{c_f}{\Gamma_f, (\epsilon_f, \delta_f)} \quad
             \hastype{\Gamma}{e}{0}{\tybool} \\\\
             \epsilon' = \max(\epsilon_t, \epsilon_f) \quad \delta' = \max(\delta_t, \delta_f)
            }
            { \tytriple
              {\Gamma}
              {\cmdif{e}{c_t}{c_f}}
              {\max(\Gamma_t, \Gamma_f), (\epsilon', \delta')}}

  \inferrule[While]
            {\tytriple{\Gamma}{c}{\Gamma, (0, 0)} \quad
             \hastype{\Gamma}{e}{0}{\tybool}
            }
            { \tytriple
              {\Gamma}
              {\cmdwhile{e}{c}}
              {\Gamma, (0, 0)}}
\end{mathparsmall}
\caption{Core Typing Rules \ifpostreview\bcp{The way epsilon and delta are
    formatted makes the rules a bit heavy.  Maybe we can experiment with
    making them subscripts or something...}\fi}
\label{fig:structural-rules1}
\end{figure}

There are three forms of assignment in \Fuzzi: (1) direct assignment to variables,
(2) indexed assignment to vectors and bags, and (3) length assignments to vectors
and bags.  There is a separate typing rule for each form of assignment
(\Cref{fig:structural-rules1}). The
\textsc{Assign} rule updates the LHS
variable's sensitivity to the derived sensitivity of the RHS
expression. The \textsc{Assign-Vector-Index} rule adds the derived sensitivity
of RHS expression to a vector's sensitivity provided the index itself is
non-sensitive. (For example, consider the vectors $xs\ag{1} = [1, 2, 3]$ and
$xs\ag{2} = [1, 2, 4]$. If we perform the assignment $xs[1] = e$ where
$e\ag{1} = 1$ and $e\ag{2} = 10$, then the two vectors become $[1, 1, 3]$ and
$[1, 10, 4]$, increasing the distance between them by $9$. We require finite
sensitivity of the vector variable on the left-hand-side to ensure  co-termination---only vectors with finite sensitivities must
have the same length.)

We have separate rules for vector and bag length updates. In
the \textsc{Assign-Vector-Length} rule, since the RHS length expression is
non-sensitive, the two vectors will be truncated or padded to the same
length. In the case of truncation, the distance between the two vectors will
decrease or remain the same; on the other hand, if both vectors were padded,
since the pad will be the same, they will not introduce any additional
distance. Thus, the LHS vector variable's sensitivity remains the same.

In the \textsc{Assign-Bag-Length} rule, it may be surprising that updating the
lengths bags---even with a $0$-sensitive new length---can result in
$\infty$-sensitivity for the bag. Consider two subsequences of the same length
$X$ and $Y$, let $L$ be their identical length. We can choose $X$ and $Y$ such
that their bag distance is $2L$. Now, the two bags $XY$ and $YX$ have distance
$0$ since they contain the same elements, but if we truncated both bags to
length $L$, then their distance grows to $2L$.  The
\textsc{Assign-Bag-Length} rule must account for this worst case scenario by
setting the sensitivity of $x$ to $\infty$.

The core typing rules for operations that involve bags are rather
restrictive. We will see in \Cref{sec:extension} how to operate more
flexibly over bags using extensions.

The \textsc{Laplace} rule computes the privacy cost of releasing a single
sensitive real value. The \text{Laplace} rule sets the sensitivity of $x$ to $0$
after noise is added, which may seem surprising since $x$'s value is
randomized. Intuitively, the $0$-sensitivity expresses that $x$'s value is now
public information and can be used in the clear. We justify $0$-sensitivity as
an upper bound on the distance between $x\ag{1}$ and $x\ag{2}$
in
\ifextended
\Cref{ap:aprhl}.
\else
Appendix B of the extended version.
\fi
However, readers do not need to look there to understand the
\Fuzzi's type sytem design.
%

The no-op command \lstinline|skip| does not alter the program state at all:
given any pre-condition $\Gamma$, we can expect the same condition to hold
after \lstinline|skip|. Also since \lstinline|skip| does not release any private
data, it has a privacy cost of $0$. This is described by the \textsc{Skip} rule.

As described in \Cref{sec:overview}, the \textsc{Sequence} rule chains together the intermediate $\Gamma$s
for two commands $c_1$ and $c_2$ and adds up the individual privacy costs for
each command.

The control flow command \lstinline|if| may modify the same variable with
different RHS expressions in each branch; if we allowed expressions with
arbitrary sensitivities as the branch condition, we would not be able to
derive valid
sensitivities for modified variables due to different execution paths. Consider
the following example, where $e$ is a sensitive boolean expression:
\lstinline|if e then x = y else x = z end|. In one execution, control flow may
follow the true branch, assigning $y$ to $x$, while, on the other execution, control
flow follows the false branch, assigning $z$ to $x$. Since the typing context
$\Gamma$ does not provide any information on the distance between $y$ and $z$,
we cannot derive a useful upper bound on $|x\ag{1} - x\ag{2}|$ after
the \lstinline|if| statement.

On the other hand, if the branch condition is a non-sensitive boolean, then we
know that control will go through the same branch in both executions. In this
case, we can take the pointwise maximum of the sensitivities from the
post-condition of both branches to derive a sound sensitivity for variables
modified by the \lstinline|if| statement. Similarly, the privacy cost of the
entire \lstinline|if| statement is bounded by the maximum of the two branches'
privacy costs.

The core typing rule for while loops require $\Gamma$ to be a loop invariant of
the loop body $c$, the loop guard $e$ a non-sensitive boolean value under
$\Gamma$, and the loop body $c$ incur no privacy cost. Had we allowed $e$ be a
sensitive boolean, then the while loop may diverge in one execution but
terminate in the other. In order to ensure that the two executions co-terminate,
we must force the values of $e\ag{1}$ and $e\ag{2}$ before each iteration. We
achieve this by checking that the invariant $\Gamma$ of the loop induces a $0$
sensitivity on the loop guard $e$.

The Simple Composition Theorem~\cite{Dwork:2006:CNS:2180286.2180305} implies
that the total privacy cost of a while loop is bounded by the sum of individual
privacy cost from each iterations. Even though we can ensure the two executions
of while loops co-terminate, we cannot always statically tell for how many
iterations both loops will run. In order to ensure the soundness of the total
privacy cost estimation, we conservatively forbid loop bodies from using any
commands that will increase $\epsilon$ or~$\delta$.

These core typing rules place rather heavy restrictions on \Fuzzi programs
operating over vectors and bags, or programs using conditionals and loops; for
example, the core rules grossly overestimate sensitivities for vectors and bags
and forbid sensitive booleans in branch and loop conditions. The core rules are
designed for chaining together blocks of differentially private mechanisms, and
these typing rules are often not enough to typecheck the implementation of
interesting differentially private algorithms.
We will see how to teach \Fuzzi's typechecker to derive more precise
sensitivities for differentially private mechanisms involving all these
constructs in
\Cref{sec:extension}.

\subsection{Soundness}
\label{sec:fuzzi-aprhl}
There has been a rich line of work on developing type systems and language
safety properties using foundational
methods~\cite[etc.]{Appel:2001:IMR:504709.504712,Appel:2007:VMM:1190216.1190235,10.1007/11693024/6,Jung:2017:RSF:3177123.3158154,Frumin:2018:RMR:3209108.3209174}. The
foundational approach develops the typing rules of a language as theorems in an
expressive logic. Type systems developed using foundational methods benefit from
the soundness of the underlying logic: if a typing rule is proven true as a
theorem, then adding new rules as theorems to the type system will not break
validity of existing rules. Most importantly, the foundational approach
allows \Fuzzi to mix typing rules for an automated typechecker with specialized
typing rules extracted from manual proofs of differential privacy.

We choose \aprhl~\cite{Barthe:2016:PDP:2933575.2934554} as the foundational
logic to build \Fuzzi's type system upon. The \aprhl logic extends Floyd-Hoare
Logic\,\cite{Hoare:1969:ABC:363235.363259} with relational assertions,
reasoning of probabilistic commands, and differential privacy cost
accounting. An \aprhl judgment has the form\ifpostreview\bcp{Could we format it so that it
looks closer to our typing judgments?}\fi{}
$\aprhlstmt{c_1}{(\epsilon, \delta)}{c_2}{\Psi}{\Phi}$.
The metavariables $c_1$ and $c_2$ stand for two programs related by this
judgment, the annotations $\epsilon$ and $\delta$ stand for the quantitative
``cost'' of establishing this relation, and $\Psi$ and $\Phi$ are both
assertions over pairs of program states, standing for the pre-condition and the
post-condition of this judgment respectively.

We have seen two kinds of rules in the \Fuzzi type system so far: expression
typing rules and core typing rules for commands. Although both are presented in
the form of inference rules, these two typing judgments are very different in
nature. The expression typing rules are defined as an inductive relation, while
the typing rules for commands are theorems to be proven. This choice is
motivated more by practicality and less by theory---foundational proofs are more
difficult to work with than inductive relations, since we do not plan on mixing
the typing rules for expressions with manual proofs, there is no need to use the
foundational methods for expressions.

Because expression typing rules of the form $\hastype{\Gamma}{e}{\phi}{\tau}$
are instances of an inductive relation, we need to prove a few soundness
properties that will make these expression typing rules useful in the
development of command typing rule proofs. In particular, we care about
soundness with respect to sensitivity and co-termination. We elide proofs by
straightforward induction.
\begin{lem}[Expression Sensitivity Sound]
\label{lem:expr-type-sens-sound}
Given $\hastype{\Gamma}{e}{\phi}{\tau}$ and two program states $M_1$ and
$M_2$ related by $\Gamma$, if $\denote{e}M_1 = v_1$ and $\denote{e}M_2 = v_2$,
then $d_\tau(v_1, v_2) \leq \phi$.
\end{lem}
%
\begin{lem}[Expression Co-termination]
\label{lem:expr-type-coterm}
Given $\hastype{\Gamma}{e}{\phi}{\tau}$ and two program states $M_1$ and
$M_2$ related by $\Gamma$, evaluating the expression $\denote{e} M_1$
yields some value $v_1$ if and only if $\denote{e} M_2$ yields some value $v_2$.
\end{lem}
%
The command typing rules have the form $\tytriple{\Gamma}{c}{\Gamma',
(\epsilon, \delta)}$. And earlier, we described $\Gamma$ and $\Gamma'$ as
pre-condition and post-conditions. What does it mean to treat a typing context
as pre- and post-conditions?

Recall the translation from typing contexts to \aprhl assertions
in \Cref{sec:overview}. This translation naturally induces a relation on each
program variable: each variable's type information $x:_\phi \tau$ becomes the
relation $d_\tau(x\ag{1}, x\ag{2}) \leq \phi$. The entire typing context
$\Gamma$ is translated to conjunctions of the pointwise relation for each
program variable. As an example, the context $x :_1 \tyint, y :_2 \tyfloat$
corresponds to the relation $d_\tyint(x\ag{1}, x\ag{2}) \leq 1 \land
d_\tyfloat(y\ag{1}, y\ag{2}) \leq 2$. We also use the $\denote{\cdot}$ function
to denote the translation of typing contexts.

So a typing rule is in fact an \aprhl judgement in disguise:
$\aprhlstmt{c}{(\epsilon, \delta)}{c}{\denote{\Gamma}}{\denote{\Gamma'}}$. To
prove these judgments valid, we need to use the \aprhl proof rules. In fact,
many of \Fuzzi's core typing rules are specialized versions of the
corresponding \aprhl rule for that command. We list all \aprhl proof rules used
in this paper in
\ifextended
\Cref{ap:aprhl},
\else
Appendix B of the extended version,
\fi
but readers do not need to look there to
understand the following content in the main body of the paper.


As an example, the soundness of the \textsc{Assign} rule is justified by the
following lemma:
\begin{lem}
Given $\hastype{\Gamma}{e}{\phi'}{\tau}$, the judgement $\aprhlstmt{x = e}{(0,
0)}{x = e}{\denote{\Gamma}}{\denote{\Gamma[x\mapsto \phi']}}$ is true.
\end{lem}
We define one such lemma for each of the typing rules given
in \Cref{fig:structural-rules1}, and justify them using corresponding \aprhl
proof rules.

One important technical subtlety is that the original presentation of \aprhl only
reasons over terminating programs. Requiring \Fuzzi's typechecker to prove
termination for all programs would
unavoidably rule out some useful ones. Fortunately, we actually need only a
subset of \aprhl's proof rules, and these are all sound even if programs only
co-terminate~\cite{cotermination:privcomm}; we can thus we relax the ``all
programs terminate'' assumption of \aprhl in the development of \Fuzzi.

\paragraph*{Remark}
Although we carry out privacy proofs in \aprhl, the logic \aprhl does not fully
isolate its user from the underlying semantics of the language. For example,
some of the \aprhl proof rules used to develop \Fuzzi require us to prove
termination of commands, but \aprhl does not give proof rules for
termination. So we develop our own sound termination typing rules that
match \aprhl's termination definition, using the semantics of \Fuzzi. This
necessitates an even lower-level logic $\mathcal{L}$ to formalize the parts not
specified by \aprhl. In the following sections, we will explicitly call out
objects defined in $\mathcal{L}$.\footnote{The proof assistant
Coq~\cite{Coq:manual} is a suitable candidate of $\mathcal{L}$; indeed, we have
already formalized some parts of \Fuzzi in Coq.}

\section{Extensions}
\label{sec:extension}
In this section, we discuss how to integrate the core \Fuzzi type system with
specialized typing rules for \Fuzzi extensions; we then introduce several
concrete extensions that will be used later for our case studies: operations
for mapping a piece of code over all the cells in a bag or vector, an
operation for partitioning a bag into a collection of smaller bags
according to some criterion, an operation for summing the elements of a bag,
and an operation for sequencing several commands using an ``advanced
composition theorem'' from the differential privacy literature to obtain a
lower privacy cost than the one given by the plain sensitivity typing rule
for sequencing.
\begin{defn}
An extension is a 4-tuple $(\mathit{ext},
  f, \mathit{rule}, \mathit{proof})$. The first field $\mathit{ext}$ is the name
  of the extension. The second field is a function $f$ that maps \Fuzzi
  expressions or commands to a \Fuzzi command, we will call $f$ the syntax
  expansion function. Let $v_1, \dots, v_i$ be the syntactic variables bound in
  $f$; the third field is a typing rule, parameterized by the same $v_1, \dots,
  v_i$ syntactic variables. The typing rule may contain premises over any
  combination of $v_1, \dots, v_i$, and the typing rule's conclusion has the
  shape of a \Fuzzi typing triple for the expanded code of the
  extension. Finally, the last field $\mathit{proof}$ is a proof of the
  soundness of the typing rule.
\end{defn}
We will use the notation $\mathit{ext}(p_1, p_2, \dots, p_i)$ for the syntax of
invoking an extension. These extension commands are replaced by the expanded
body $f\,(p_1,\, p_2,\, \dots,\, p_i)$ with each $p_i$ substituting for each
syntax variable $v_i$.

Expressing the typing rule requires the full generality of the lowest level
logic $\mathcal{L}$, since extension typing rules contain side conditions (such
as termination) that are not captured by \aprhl. The general shape of this
theorem is of the form
$$\forall\, v_1,\dots,v_i,\, P_1 \land \dots \land
P_j \Rightarrow \tytriple{\Gamma}{\mathit{ext}(v_1, \dots, v_i)}{\Gamma',
(\epsilon, \delta)}$$ Each $P_j$ is a premise of the typing rule, and all
premises may bind any combination of $v_1, \dots, v_i$. The conclusion of the
theorem is always of the shape of a \Fuzzi typing triple, so that proofs of the
typing rule can mix with the soundness proofs of the core typing rules
introduced in \Cref{sec:typing-fuzzi}. The final component of an extension is a
soundness proof of the typing rule theorem.

Some of these premises are \Fuzzi typing judgments, while some others are
auxiliary judgments that asserts termination or describes a linear scaling
relationship between the pre-condition sensitivities and the post-condition
sensitivities. These two extra kinds of auxiliary judgments are defined in
$\mathcal{L}$, we will give definitions for these auxiliary judgments as we
encounter them. We will describe their proof rules in
\ifextended
\Cref{ap:aux-checker}.
\else
Appendix C of the extended version.
\fi

We will provide an overview of the use case for each extension, and only provide
a sketch of the proof of soundness to conserve space. Detailed soundness proofs
for all the extensions can be found in
\ifextended
\Cref{ap:proofs}.
\else
Appendix D of the extended version.
\fi
\begin{figure}[t]
  \begin{minipage}{0.45\textwidth}
  \centering
  \begin{mathparsmall}
    \mprset{flushleft}
    \inferrule[Bag-Map]
              {\tinytext{Termination, Deterministic}
               \\\\ \term{\Gamma}{c}
               \qquad \determ\; c
               \\\\ \tinytext{Should Not Modify}
               \\\\ \mathit{t_{in}}, \mathit{in},\mathit{out},\mathit{i} \notin \mv\; c
               \\\\ \tinytext{Abbreviation}
               \\\\ \sigma\hspace{0.3em} = [\mv\, c, i, in, out \mapsto \infty]
               \\\\ \sigma' = [\mv\,c, i, \mathit{t_{in}}, \mathit{t_{out}} \mapsto \infty]
               \\\\ \tinytext{Dependency}
               \\\\ \tytriple
               {\mathtt{stretch}\;\Gamma[\mathit{t_{in}} \mapsto \hspace{0.38em}0]\sigma}
               {c}
               {\Gamma_1, (0, 0)}
               \\\\ \tytriple{\mathtt{stretch}\;\Gamma[\mathit{t_{in} \mapsto \infty}]\sigma}
                     {c}
                     {\Gamma_2, (0, 0)}
               \\\\ \Gamma_1(t_\mathit{out}) = 0
               \\\\ \{x \mathrel{|} x\in\mv\; c \land \Gamma_2(x) > 0\} \subseteq \{t_{\mathit{out}}\}
               \\\\ \tinytext{Output Sensitivity}
               \\\\ \Gamma_1(\mathit{t_{out}}) = 0
               \quad \Gamma_{\mathit{out}} =
               \Gamma
               [\mathit{out} \mapsto \Gamma(\mathit{in})]\sigma'
               }
              {\tytriple
                {\Gamma}
                {\cmdbmap{\mathit{in}}{\mathit{out}}%
                         {t_\mathit{in}}{i}{t_\mathit{out}}{c}}
                {\Gamma_{\mathit{out}}, (0, 0)}}
  \end{mathparsmall}
  \end{minipage}
  \begin{minipage}{0.45\textwidth}
  \begin{lstlisting}[xleftmargin=3em]
    &$\mathit{i}$& = 0;
    &$\mathit{out}$&.length = &$\mathit{in}$&.length;
    while &$\mathit{i}$& < &$\mathit{in}$&.length do
      &$t_\mathit{in}$& = &$\mathit{in}$&[&$\mathit{i}$&];
      &$\mathit{c}$&;
      &$\mathit{out}$&[&$\mathit{i}$&] = &$t_\mathit{out}$&;
      &$\mathit{i}$& = &$\mathit{i}$& + 1;
    end
  \end{lstlisting}
  \end{minipage}
  \caption{\textsc{Bag-Map} typing rule and extension code pattern}
  \label{fig:bag-map}
\end{figure}
\subsection{Bag Map}

Our first extension, \textsc{Bag-Map}, takes an input bag variable,
an output bag variable, a few auxiliary variables used by the expanded loop, and
finally a ``bag-map body'' $c$ that reads from a single bag entry and outputs a
mapped value for that bag entry. This command $c$ represents a single step of
the ``map'' operation. \textsc{Bag-Map} applies this operation uniformly for all
entries in a bag.

The premises of \textsc{Bag-Map}'s typing rule use a few new ingredients---the
function $\mv$ collects the set of modified variables from a command $c$, and
the function $\mathtt{stretch}$ ``expands'' a typing context. The
$\mathtt{stretch}$ function takes a typing context $\Gamma$, and for each
variable $x \in \Gamma$, if $\Gamma(x) > 0$, sets $\Gamma(x) = \infty$,
otherwise leaves $\Gamma(x)$ as $0$.
\begin{center}
\begin{tabular}{RLCLL}
\mathtt{stretch}\, &\varnothing &=& \varnothing \\
\mathtt{stretch}\, &(x :_0 \tau, \Gamma) &=& x :_0 \tau, \mathtt{stretch}\,\, \Gamma \\
\mathtt{stretch}\, &(x :_\sigma \tau, \Gamma) &=& x :_\infty \tau,
\mathtt{stretch}\,\, \Gamma & \mathrm{if}\ \sigma>0
\end{tabular}
\end{center}
The function $\mathtt{stretch}$ is used when we need to verify that some
output variable's sensitive value derives solely from some input
variable.%
\footnote{This kind of dependency analysis is one of
the motivating examples for Benton's seminal work on
relational Hoare logic\,\cite{Benton:2004:SRC:964001.964003}.}

Since a variable can only become sensitive if its value derives from another
sensitive variable, it is perhaps not surprising that \Fuzzi's sensitivity type
system is capable of tracking dependency as well. We will use the function
$\mathtt{stretch}$ to help uncover this dependency analysis part of the
sensitivity type system.

Consider any program fragment $c$, for which we want to verify that for a single
variable $t$, after executing $c$, the sensitive data held in $t$ must only come
from $s$. If we had the following typing judgment about $c$:
$\tytriple{\mathtt{stretch}\, \Gamma_1[s\mapsto 0]}{c}{\Gamma_2, (0, 0)}$ where
$\Gamma_2(t) = 0$, then we know if $s$ is non-sensitive, $t$ is also
non-sensitive. This implies that the only sensitive dependency of $t$ is at most
the singleton set $\{s\}$.

However, this typing judgment only tells us the dependency of $t$ when all
sensitive variables are $\infty$-sensitive before executing $c$. Does the same
result hold when those variables have finite sensitivity? To arrive at this
conclusion, we need to apply the \textsc{Conseq} rule from \aprhl after
unfolding our previous typing judgment into an \aprhl judgement.
\begin{mathparsmall}
\inferrule[Conseq]
          {\aprhlstmt{c_1}{\epsilon', \delta'}{c_2}{\Phi'}{\Psi'} \quad
           \vDash \Phi \Rightarrow \Phi' \\\\
           \vDash \Psi' \Rightarrow \Psi \quad
           \vDash \epsilon' \leq \epsilon \quad \vDash \delta' \leq \delta}
          {\aprhlstmt{c_1}{(\epsilon, \delta)}{c_2}{\Phi}{\Psi}}
\end{mathparsmall}%
The \textsc{Conseq} rule allows us to strengthen the pre-condition. Now, in the
 context $\mathtt{stretch}\Gamma_1[s\mapsto 0]$, if we change any other
variable's sensitivity to some finite value, then we have a stronger statement,
since $\infty$-sensitivity is implied by finite sensitivity.  Thus,
the \textsc{Conseq} rule allows us to change the pre-condition to one that
implies the stretched typing context.

This technique allows us to verify $t$'s dependency is at most $\{s\}$ through
sensitivity analysis. However, the program $c$ may modify variables other than
$t$. In order to make sure there are no other sensitive output variables from
the program fragment $c$, we also want to verify the lack of dependency on
$s$. Can we re-use the sensitivity type system to check some other modified
variable $v$ does not depend on the variable $s$?  Indeed we can, using the
$\mathtt{stretch}$ function again in a slightly different way. Consider the
typing judgment
$\tytriple{\mathtt{stretch}\,\Gamma_1[s \mapsto \infty]}{c}{\Gamma_2, (0, 0)}$ where
$\Gamma_2(v) = 0$. This typing judgment tells us that no matter how much $s$
and the other sensitive values changes between two executions of $c$, the value
held in $v$ remains the same at the end of the execution. So indeed $v$ does not
depend on $s$ or any other sensitive variable. Again \textsc{Conseq} rule allows
us to strengthen the $\infty$-sensitivity to any finite sensitivity.

The \textsc{Bag-Map} typing rule applies this technique to check that, on each
iteration of $c$, the value of $t_\mathit{out}$ derives only from the
corresponding input bag entry $t_\mathit{in}$. The program fragment $c$ should
not access the original bag value directly, and neither should it write directly
to the output bag. Furthermore, each iteration of the bag map body should be
independent of each other, so the values of its modified variables should not
carry over to the next iteration. For these reasons, we set the sensitivity of
modified variables of $c$, the variable $i$, and the input and outputs bags
variables $\mathit{in}$ and $\mathit{out}$ to $\infty$ in the judgment
$\tytriple{\mathtt{stretch}\;\Gamma[\mathit{t_{in}} \mapsto 0]\sigma} {c}
{\Gamma_1, (0, 0)}$, and check $\Gamma_1(t_\mathit{out}) = 0$. We use the letter
$\sigma$ to abbreviate the update expression $[\mv\;
c,i, \mathit{in}, \mathit{out} \mapsto \infty]$.

We also want to verify that none of the modified variables, except for
$t_\mathit{out}$, has any dependency on sensitive data. This is why we check
the only variable that can potentially hold sensitive data is $t_\mathit{out}$
in the judgment
$\tytriple{\mathtt{stretch}\;\Gamma[\mathit{t_{in}} \mapsto \infty]\sigma} {c}
           {\Gamma_2, (0, 0)}$
The variables $\mv\; c, i, \mathit{in}, \mathit{out}$ are again set to $\infty$
to ensure these variables do not leak information across iterations as discussed
above.

We use the judgment $\determ\, c$ to assert that $c$ is a deterministic
\Fuzzi program. It is easy to show that any program $c$ that do not contain
the Laplace mechanism is deterministic. We use the judgment $\term{\Gamma}{c}$
to denote that command $c$ terminates for any program state that is well-typed
according to $\shape(\Gamma)$.  That is, for all $M$ in $\shape(\Gamma)$,
running the program $c$ with the well-typed program state $\denote{c}M$
terminates with probability~$1$.

Since the \aprhl judgment that corresponds to the conclusion of this typing
rule implies co-termination of \textsc{Bag-Map} programs, we need to prove the
expanded while loops actually co-terminate. However, because these two while
loops may have different number of iterations due to bags having different
sizes, thus executing $c$ for different number of times, we cannot simply show
$c$ co-terminates. So, we take the extra step of requiring termination of $c$ on
all well-shaped inputs, which ensures both loops will always terminate.

The soundness proof for the \textsc{Bag-Map} typing rule applies
dependency analysis
to ensure the map body $c$ maps the input value
$t_\mathit{in}$ deterministically to $t_\mathit{out}$, and that $c$ does not
store sensitive data in any of its other modified variables. This allows us to
ensure $in[i]$ maps deterministically and uniformly to $out[i]$ through each
iteration of $c$. Now, adding or removing any entry from the input bag will
correspondingly add or remove the mapped value from the output bag. So the
output bag must have the same sensitivity as the input bag.

\subsection{Vector Map}
Our second example, the \textsc{Vector-Map} extension is, very similar
to \textsc{Bag-Map} in that it also requires the ``map'' command to restrict its flow
of sensitive data from $t_\mathit{in}$ to only
$t_\mathit{out}$. However, \textsc{Vector-Map} has an additional requirement
that map body must be ``linear'':
\begin{defn}[Linear Commands]
\label{defn:linear-commands}
We write $k \Gamma$ to denote a typing context $\Gamma'$ where $\Gamma'(x) =
k \Gamma(x)$. A deterministic and terminating command $c$ is linear with respect
to $\Gamma_1$ and $\Gamma_2$, if for any $k > 0$, the scaled typing judgment
$\tytriple{k\Gamma_1}{c}{k\Gamma_2, (0, 0)}$ is true. We define $k\cdot \infty
= \infty$ for $k > 0$, and $0\cdot \infty = 0$.
\end{defn}
This definition tells us that the updates in sensitivity in the post-condition
scale linearly with respect to the sensitivities in the pre-condition. An
example of a linear command is $x = 2 y + 1$ with typing context $\Gamma_1 =
x:_0 \tyfloat, y :_1 \tyfloat$ and $\Gamma_2 = x:_2 \tyfloat, y :_1 \tyfloat$.
\begin{figure}[t]
  \begin{minipage}{0.45\textwidth}
  \centering
  \begin{mathparsmall}
    \mprset{flushleft}
    \inferrule[Vector-Map]
              {\tinytext{Termination, Deterministic}
               \\\\ \term{\Gamma}{c}
               \qquad \determ\, c
               \\\\ \tinytext{Should Not Modify}
               \\\\ \mathit{t_{in}}, \mathit{in},\mathit{out},\mathit{i} \notin \mv\, c
               \\\\ \tinytext{Abbreviation}
               \\\\ \sigma\hspace{0.3em} = [\mv\, c, i, in, out \mapsto \infty]
               \\\\ \sigma' = [\mv\, c, i, \mathit{t_{in}}, \mathit{t_{out}} \mapsto \infty]
               \\\\ \tinytext{Dependency}
               \\\\ \tytriple
               {\mathtt{stretch}\,\Gamma[\mathit{t_{in}} \mapsto \hspace{0.38em}0]\sigma}
               {c}
               {\Gamma_1, (0, 0)}
               \\\\ \tytriple{\mathtt{stretch}\,\Gamma[\mathit{t_{in}} \mapsto \infty]\sigma}
                     {c}
                     {\Gamma_2, (0, 0)}
               \\\\ \Gamma_1(t_\mathit{out}) = 0
               \\\\ \{x | x\in\mv\, c \land \Gamma_2(x) > 0\} \subseteq \{t_{\mathit{out}}\}
               \\\\ \tinytext{Linear}
               \\\\ \tytriple{\Gamma[\mathit{t_{in}}\mapsto 1]\sigma}
                     {c}
                     {\Gamma_3, (0, 0)}\,\linear
               \\\\ \tinytext{Output Sensitivity}
               \\\\ \Gamma_{\mathit{out}} =
               \Gamma[\mathit{out} \mapsto \Gamma(in) \cdot \Gamma_3(\mathit{t_{out}})]\sigma'
               }
              {\tytriple
                {\Gamma}
                {\cmdamap{\mathit{in}}{\mathit{out}}{t_\mathit{in}}{i}{t_\mathit{out}}{c}}
                {\Gamma_{\mathit{out}}, (0, 0)}}
  \end{mathparsmall}
  \end{minipage}
  \begin{minipage}{0.45\textwidth}
  \begin{lstlisting}[xleftmargin=2.5em]
    &$\mathit{i}$& = 0;
    &$\mathit{out}$&.length = &$\mathit{in}$&.length;
    while &$\mathit{i}$& < &$\mathit{in}$&.length do
      &$t_\mathit{in}$& = &$\mathit{in}$&[&$\mathit{i}$&];
      &$\mathit{c}$&;
      &$\mathit{out}$&[&$\mathit{i}$&] = &$t_\mathit{out}$&;
      &$\mathit{i}$& = &$\mathit{i}$& + 1;
    end
  \end{lstlisting}
  \end{minipage}
  \caption{\textsc{Vector-Map} typing rule and extension code pattern}
  \label{fig:array-map}
\end{figure}

A counterexample is $\cmdif{x > 0}{x = x + 1}{x = x + 2}$ with $\Gamma_1 = x
:_1 \tyfloat$ and $\Gamma_2 = x:_2 \tyfloat$. This command conditionally
increments $x$ by a constant of $1$ or $2$, so if $x$ was $1$ sensitive before
executing this command, then we can show in \aprhl that $x$ is $2$ sensitive
after this conditional command. Now, if we scaled $x$'s sensitivity by $0.5$,
then $x$ is $1.5$ sensitive after this conditional increment, rather than
$2 \cdot 0.5 = 1$. So this command is not linear with respect to the chose
$\Gamma_1$ and $\Gamma_2$. However, had we chosen $\Gamma_2 =
x:_\infty \tyfloat$, then this command is linear with respect to the new
post-condition, because $k\cdot \infty = \infty$, and no matter what the values
of $x\ag{1}$ and $x\ag{2}$ are, their difference is always bounded by
$\infty$. We present the proof rules for linear commands
in
\ifextended
\Cref{ap:aux-checker}.
\else
Appendix C of the extended version.
\fi

For vector map, we need to know the scaling relationship between the sensitivity
of $t_\mathit{out}$ and the sensitivity of $t_\mathit{in}$ in order to derive
the sensitivity of the output vector. With $c$ being a linear command for the
chosen pre-condition $\Gamma[\mathit{t_{in}} \mapsto 1]\sigma$ and
post-condition $\Gamma_3$, by the definition, for any scale factor $k > 0$ we
know if $d(t_\mathit{in}\ag{1}, t_\mathit{in}\ag{2}) \leq k$ before executing
$c$, then $d(t_\mathit{out}\ag{1}, t_\mathit{out}\ag{2}) \leq sk$ after
executing $c$, where $s = \Gamma_3(t_\mathit{out})$. Recall the definition of
vector distance, by instantiating $k$ with the actual distance for each pair of
$i$th entries from the input vectors
$d(\mathit{in}\ag{1}[i], \mathit{in}\ag{2}[i])$, we know the distance between
the output vectors satisfy the following condition:
\begin{align*}
d_{[\tau]}(\mathit{out}\ag{1}, \mathit{out}\ag{2}) &= \sum_i d_\tau(\mathit{out}\ag{1}[i],
 \mathit{out}\ag{2}[i]) \\
 &\leq \sum_i s\, d_\tau(\mathit{in}\ag{1}[i],\mathit{in}\ag{2}[i]) \\
 &= s d_{[\tau]}(\mathit{in}\ag{1}, \mathit{in}\ag{2}).
\end{align*}
This justifies the sensitivity derived by the typing rule for vector map.

\subsection{Partition}
Our third extension, \textsc{Partition}, allows programmers to break apart a
larger bag into a vector of smaller bags. \textsc{Partition} is parameterized by
an input bag variable, an output vector variable, a few auxiliary variables for
storing results from intermediate computations, and finally a command that maps
each input bag entry to a partition index.

The \textsc{Partition} extension is similar to \textsc{Bag-Map} in that it maps
each bag item to some value, but they differ in how the output value from the
each iteration is used. With \textsc{Bag-Map}, the output value from each
iteration is collected into the output bag as is. \textsc{Partition} uses the
output value as an assignment index into the output bag $\mathit{out}$, and
appends the bag entry at current iteration to the sub-bag at
$\mathit{out}[t_\mathit{idx}]$. As an example, if the input bag is $[1.2, 2.3,
3.4]$, and the map operation simply rounds down each value to the nearest
integer, then the output bag will be $\left[[], [1.2], [2.3], [3.4]\right]$.

\begin{figure}[t]
  \begin{mathparsmall}
    \mprset{flushleft}
    \inferrule[Partition]
              {
               \tinytext{Termination, Deterministic}
               \\\\ \term{\Gamma}{c} \quad \determ\, c
               \\\\ \tinytext{Should Not Modify}
               \\\\ \mathit{t_{in}}, \mathit{in}, \mathit{out}, i, \mathit{out_{idx}} \notin \mv\, c
               \\\\ \tinytext{Abbreviation}
               \\\\ \sigma\hspace{0.3em} = [\mv\, c, i, in, \mathit{out_{idx}} \mapsto \infty]
               \\\\ \sigma' = [\mv\, c, i,
                               \mathit{t_{in}}, \mathit{t_{idx}}, \mathit{out_{idx}},
                               \mathit{t_{part}} \mapsto \infty]
               \\\\ \tinytext{Number of Partitions Non-Sensitive}
               \\\\ \hastype{\Gamma}{\mathit{nParts}}{0}{\tyint}
               \\\\ \fv\, \mathit{nParts} \cap \mv\, \mathtt{c} = \varnothing
               \\\\ i, \mathit{t_{in}}, \mathit{t_{idx}}, \mathit{out_{idx}}, \mathit{t_{part}} \notin \fv\,\mathit{nParts}
               \\\\ \tinytext{Dependency}
               \\\\
               \tytriple
               {\mathtt{stretch}\,\Gamma[\mathit{t_{in}} \mapsto \hspace{0.38em}0]\sigma}
               {c}
               {\Gamma_1, 0}
               \\\\ \tytriple{\mathtt{stretch}\,\Gamma[\mathit{t_{in}} \mapsto \infty]\sigma}
                     {c}
                     {\Gamma_2, 0}
               \\\\ \Gamma_1(\mathit{t_{out}}) = 0
               \\\\ \{x | x\in\mv\, c \land \Gamma_2(x) > 0\} \subseteq \{t_{\mathit{out}}\}
               \\\\ \tinytext{Output Sensitivity}
               \\\\ \Gamma_{\mathit{out}} =
               \Gamma[\mathit{out} \mapsto \Gamma(\mathit{in})]
                     \sigma'
              }
              {\tytriple
                {\Gamma}
                {\bm{\mathtt{partition}}
                  (\mathit{in}, \mathit{out},
                  t_\mathit{in}, i, \mathit{t_{out}},
                  \mathit{t_{idx}}, \mathit{out_{idx}},
                  \mathit{t_{part}}, \mathit{nParts}, c)}
                {\Gamma_{\mathit{out}}, 0}}
  \end{mathparsmall}
  \begin{minipage}{0.45\textwidth}
  \begin{lstlisting}[xleftmargin=0em]
    &$i$& = 0;
    &$\mathit{out}$&.length = &$\mathit{nParts}$&;
    while &$i$& < &$\mathit{nParts}$& do
      &$\mathit{out}$&[&$i$&].length = 0;
      &$i$& = &$i$& + 1;
    end;
    bmap(&$\mathit{in}$&, &$\mathit{out_{idx}}$&, &$\mathit{t_{in}}$&, &$i$&, &$\mathit{t_{out}}$&, c);
    &$i$& = 0;
  \end{lstlisting}
  \end{minipage}
  \begin{minipage}{0.50\textwidth}
  \begin{lstlisting}[xleftmargin=-2em]
    while &$i$& < &$\mathit{out_{idx}}$&.length do
      &$\mathit{t_{idx}}$& = &$\mathit{out_{idx}}$&[&$i$&];
      if 0 <= &$t_\mathit{idx}$& &\verb|&&|& &$t_\mathit{idx}$& < &$\mathit{out}$&.length then
        &$\mathit{t_{part}}$& = &$\mathit{out}$&[&$t_\mathit{idx}$&];
        &$\mathit{t_{part}}$&.length = &$\mathit{t_{part}}$&.length + 1;
        &$\mathit{t_{part}}$&[&$\mathit{t_{part}}$&.length - 1] = &$\mathit{in}$&[&$i$&];
        &$\mathit{out}$&[&$t_{\mathit{idx}}$&] = &$\mathit{t_{part}}$&;
      else
        skip;
      end;
      &$i$& = &$i$& + 1;
    end
  \end{lstlisting}
  \end{minipage}
  \caption{\textsc{Partition} typing rule and extension expansion}
  \label{fig:partition}
\end{figure}

It may seem redundant that \textsc{Partition} takes the number of partitions as
a parameter. Shouldn't \textsc{Partition} be able to compute the number of
partitions as it processes the input bag values? It should not, because a
computed number of
partitions is a sensitive value that depends on the contents of the
input bag. Taking the previous example, if we add a value of $100.1$ to the input
bag and do not fix the number of partitions ahead of time, then the output vector
will have $97$ more sub-bags than the original output vector---i.e., the distance
between the output vectors can be made arbitrarily large by adding a single item
to the input bag.
This is why we fix the number of partitions
and drop the items whose partition indices are out of range.

The soundness of
the sensitivity check for partition comes from the fact that each index is
derived
only from its corresponding bag entry, thus adding or removing one bag entry can
cause at most one sub-bag in the output vector to vary by distance
$1$. Generalizing this fact shows that the output vector has the same sensitivity as
the input bag does to partition.

\begin{figure}[t]
\begin{minipage}{0.45\textwidth}
\begin{mathparsmall}
\inferrule[Bag-Sum]
          {\literal\,\mathit{bound} \quad \mathit{bound} \geq 0\\\\
           \phi = \Gamma(\mathit{in}) \quad \Gamma_\mathit{out}
           = \Gamma[\mathit{out} \mapsto \phi\cdot \mathit{bound}][i, \mathit{t_{in}} \mapsto \infty]
          }
          {\tytriple
            {\Gamma}
            {\bm{\mathtt{bsum}} (\mathit{in}, \mathit{out}, i, t_\mathit{in}, \mathit{bound})}
            {\Gamma_\mathit{out}, 0}
          }
\end{mathparsmall}
\end{minipage}
\begin{minipage}{0.45\textwidth}
\begin{lstlisting}
i = 0;
out = 0;
while i < in.length do
  &$t_\mathit{in}$& = in[i];
  if &$t_\mathit{in}$& < &$-\mathit{bound}$& then
    out = out - &$\mathit{bound}$&;
  else
    if &$t_\mathit{in}$& > &$\mathit{bound}$& then
      out = out + &$\mathit{bound}$&;
    else
      out = out + &$t_\mathit{in}$&;
    end
  end
  i = i + 1;
end
\end{lstlisting}
\end{minipage}
\caption{\textsc{Bag-Sum} typing rule and expansion}
\end{figure}
\subsection{Bag Sum}
Our fourth extension, \textsc{Bag-Sum}, works with bags of real-valued
data and adds these values up with clipping. The clipping process truncates a
value $s$ such that its magnitude is no larger than $\mathit{bound}$. This is
important to ensure the output of \textsc{Bag-Sum} has finite
sensitivity. Recall that the sensitivity definition on a bag places no
constraints on the distance of values held by the bag. If we na\"ively summed
the two bags $[1, 2]$ and $[1, 2, 100]$, although their bag distance is
bounded by $1$, their sums have distance $100$. Using only the bag distance, the typechecker will have no information
on the sensitivity of the sum. Truncating each value
into the range $[-\mathit{bound}, \mathit{bound}]$ allows us to bound the total
sensitivity of the sum value---if up to $\phi$ bag items may be added or
removed, and each can contribute up to $\mathit{bound}$ towards the total
sensitivity of \lstinline|out|, then at the end of the loop, \lstinline|out| must
be $(\phi\cdot \mathit{bound})$-sensitive.
\begin{figure}[t]
\begin{minipage}{0.45\textwidth}
\begin{mathparsmall}
\mprset{flushleft}
\inferrule[Adv-Comp]
          {\tinytext{Loop body}
           \\\\ \tytriple{\Gamma}{c}{\Gamma, (\epsilon, \delta)}
           \\\\ \tinytext{Privacy cost}
           \\\\ \epsilon^* = \epsilon\sqrt{2n \ln (1/\omega)} + n\epsilon(e^\epsilon-1)
           \\\\ \delta^* = n\delta + \omega
           \\\\ \tinytext{Should Not Modify}
           \\\\ i \notin \mv\, c
           \\\\ \tinytext{{Adv-Comp} Parameters}
           \\\\ \omega > 0
           \qquad n > 0
           \\\\ \literal\, \omega
           \qquad \literal\, n}
          {\tytriple{\Gamma}{\bm{\mathtt{ac}}(i, n, \omega, c)}{\Gamma, (\epsilon^*, \delta^*)}}
\end{mathparsmall}
\end{minipage}
\begin{minipage}{0.45\textwidth}
\begin{lstlisting}
i = 0;
while i < &$n$& do
  &$c$&;
  i = i + 1;
end
\end{lstlisting}
\end{minipage}
\caption{\textsc{Adv-Comp} typing rule and expansion}
\end{figure}

\subsection{Advanced Composition}
Our fifth extension, \textsc{Adv-Comp}, simply expands to a loop that runs
the supplied command $c$ for $n$ times. However, this extension provides a
special privacy cost accounting mechanism known as Advanced
Composition~\cite{DworkRoVa10}. Compared to Simple
Composition, \textsc{Adv-Comp} gives an asymptotically better $\epsilon$ that
grows at the rate of $O(\sqrt{n})$, at the cost of a small increase in
$\delta$. Simple composition of loop iterations will give privacy costs that
grow at the rate of $O(n)$ instead. The programmer chooses the increase in
$\delta$ by providing a positive real number $\omega$, which is used to compute
the aggregated privacy cost for the entire loop.

The \textsc{Adv-Comp} extension is useful for programs that iteratively release
 data, and it also allows programs to be run for more iterations
 while staying under the same privacy budget. We use \textsc{Adv-Comp} in our
 implementation of logistic regression in \Cref{sec:eval}.

It is worth noting that \textsc{Adv-Comp} does not always give a better privacy
cost than simple composition: when the $\epsilon$ cost of $c$ is large, the term
$n\epsilon(e^\epsilon-1)$ becomes the dominating term. This term again grows
linearly with $n$, and it has a multiplicative factor of $e^\epsilon -
1$. When \textsc{Adv-Comp} gives worse privacy cost in both $\epsilon$ and
$\delta$ in comparison to simple composition, the type system falls back to
simple composition for the expanded loop for privacy cost accounting.

\section{Implementation}
\label{sec:impl}
Our prototype \Fuzzi checker expands extensions before typechecking, leaving
hints in the expanded abstract syntax tree so that it can tell when to apply the
macro-typing rules that accompany each extension. The typechecking algorithm
includes three major components: 1) a checker that computes sensitivity, 2) a
checker for termination, and 3) a checker for linear properties\bcp{not sure
what this means} of commands.

The implementation uses three separate ASTs
types---called \textsc{Imp}, \textsc{ImpExt}, and \textsc{ImpTC}---to represent
a \Fuzzi program in different phases of checking.  The \textsc{Imp} AST is what
the parser produces---i.e., the constructs of the core language plus extension
applications.  \textsc{ImpExt} is a convenient language for entering extension
declarations for \Fuzzi.
Finally, \textsc{ImpTC} represents programs expanded from \textsc{Imp}
---while-language with no extension application nor extension definition, but
contains typechecker hints, and the typechecker expects terms
from \textsc{ImpTC}. The \textsc{ImpTC} language is not accepted by the parser,
as we do not anticipate users wanting to enter typechecker hints directly. We
use the extensible sum encoding described in \textit{data types \`{a} la carte}
~\cite{Swierstra:2008:DTL:1394794.1394795} to represent these ASTs in order to
avoid code duplication. We depend on the \texttt{compdata}
package~\cite{Bahr:2011:CDT:2036918.2036930} to manipulate ASTs in this
encoding.

The three checkers are implemented separately. A checker composition function
takes results from each checker, and produces the final type information for
a \Fuzzi program.

To efficiently execute \Fuzzi code, we compile \Fuzzi programs to
Python and use fast numeric operations from the \verb|numpy|
library~\cite{Oliphant:2015:GN:2886196} whenever appropriate.

\section{Evaluation}
\label{sec:eval}
To evaluate \Fuzzi's effectiveness, we implement four differentially private
learning algorithms from four diverse classes of learning
methods---discriminative models, ensemble models, generative models, and
instance-based learning.  The algorithms and datasets are both taken from
canonical sources. We want to know (1) whether \Fuzzi can express these
algorithms adequately, (2) whether the typechecker derives sensitivity
bounds comparable to results of a careful manual analysis, and (3) whether
the final privacy costs are within a reasonable range.

We use datasets obtained from the UCI Machine Learning repository
\cite{Dua:2017} and the MNIST database of handwritten digits
\cite{lecun-mnisthandwrittendigit-2010}. We focus on evaluating \Fuzzi's
usability on prototyping differentially private learning tasks in these
experiments, rather than trying to achieve state-of-the-art learning
performance.

We find that \Fuzzi can indeed express all four examples and that it
correctly derives sensitivity bounds comparable to results from a manual
analysis.  The examples also demonstrate that the extensions described in
\Cref{sec:extension} are useful for real-world differential privacy
programming, since each of the learning algorithms can be expressed as a
straightforward combination of extensions.

On the other hand, the privacy costs that \Fuzzi derives are arguably a bit
disappointing.  One reason for this is that we ran the experiments on fairly
small datasets.  A deeper reason is that \Fuzzi focuses on accurate
automatic inference of \textit{sensitivities}, an important building block
in differential privacy. Tracking sensitivities is somewhat orthogonal to the question
of how to most tightly track {privacy costs}, which is achieved via composition theorems that
sit on top of the sensitivity calculations. Our focus in this work has been
mainly on tracking sensitivity; in particular, we implement only simple
composition theorem in the core type system. The result is that \Fuzzi may
report a larger privacy cost than is optimal, even when it optimally
computes sensitivities. However, stronger composition theorems can be added
as extensions: we give an example of this by demonstrating an ``advanced
composition'' \cite{DworkRoVa10} extension in \Cref{sec:extension}. We view
adding extensions for more sophisticated methods of tracking privacy costs
as future work (see \Cref{sec:related-works} and \Cref{sec:future}).

\ifpostreview\bcp{Still not 100\% happy with all this.  In particular, I
  don't understand why we could not run similar experiments with more
  data and get better numbers...\fi

\subsection{Logistic Regression}
\label{sec:eval-logistic-regression}
We first investigate a binary classification problem. The dataset contains 12,665
digits (either $0$ or $1$) from the MNIST database. (We only work with $0$ and
$1$ digits because it simplifies our presentation. A $10$-class logistic
regression model that classifies all $10$ digits can be implemented using the
same methods we show here.)
We use 11,665 digits for training and leave 1,000 digits on the side for
evaluation. Each digit is represented by a $28\times 28$ grayscale image plus a
label indicating whether it is a $0$ or a $1$. The image and its label are
flattened into a $785$-dimensional vector. We then use gradient descent, a
simple and common machine learning technique to train a standard logistic
regression model that classifies these two shapes of digits. We apply
differential privacy here to protect the privacy of each individual image of the
digits. In other words, differential privacy limits an adversary's ability to
tell whether a particular image was used in training the classification model.
In particular, we modify gradient descent with a gradient clipping step to
achieve differential privacy. Gradient clipping is a common technique for
implementing differentially private gradient
descent~\cite{Abadi:2016:DLD:2976749.2978318,brendan2018learning}.

The logistic regression model is parameterized by a vector $\vec{w}$ of the same
dimension as the input data and a scalar $b$. A ``loss function'' $L(\vec{w},
b, \vec{x}_i)$ quantifies the mistakes the model makes given a pair of $\vec{w}$
and $b$, and an input image $x_i$. In ordinary (non-private) gradient
descent, we compute
the gradients $\frac{\partial\,L}{\partial\,\vec{w}}$ and
$\frac{\partial\,L}{\partial\,b}$ for each image $x_i$, and we move the current
parameters $\vec{w}$ and $b$ in the direction of the average of these gradients, decreasing
the value of the loss function $L$ (i.e., improving the quality of model parameters). To set the
initial values of $\vec{w}$ and $b$, we take random samples from a normal
distribution centered at $0$ and with variance $1$.

Since the gradients here are computed from private images, the model parameters
modified with these gradients are also sensitive information that cannot be
released directly. Instead, we release noised
estimations of the average gradients and use these values to update model
parameters. We apply \lstinline|bmap| over the input dataset, computing a bag of
both gradients for each image. We then use \lstinline|bsum| and the Laplace
mechanism to release the sum of the gradients. We also use the Laplace mechanism
to compute a noised estimate of the size of the dataset, and then update the model parameter with the
noised average gradient. The \lstinline|bsum| extension clips each gradient
value so that the final sum has a bounded sensitivity.

We iterate the gradient descent calculation
with the \lstinline|ac| (advanced
composition) extension. With $100$ passes over the training set, we reach a
training accuracy of $0.933$, and test accuracy of $0.84$. We measure accuracy as the fraction of images the trained model correctly classifies. The differentially private model's accuracy is comparable to the
accuracy of $0.88$ for a logistic regression model without differential
privacy~\cite{726791}. Training the model with $100$ passes incurs privacy cost
$\epsilon = 11.02$ and $\delta = 10^{-6}$.  Our $\epsilon$ privacy cost is larger than the results
achieved by \citet{Abadi:2016:DLD:2976749.2978318} ($\epsilon = 2.55, \delta =
10^{-5}$) on MNIST, due to our use of a simpler privacy composition theorem;
Abadi et al. invented a specialized ``moments accountant'' method to derive
tighter aggregated privacy costs given the same sensitivity analysis.
%
%
\subsection{Teacher Ensemble}
\label{sec:eval-teacher-ensemble}
Next, we build on the logistic regression from above, together with ideas
from \citet{papernot:semi-supervised:2016}, to design an {\em ensemble
  model}---a
collection of models---that classifies $0$ and $1$ digits from the MNIST
dataset. Papernot et al. demonstrated a general approach for providing
privacy guarantees with respect to the training data: first partition the
private training data into $k$ partitions, then apply a training procedure
to each partition to build a private model for it.
These $k$ private models form a Private Aggregation of Teacher Ensembles
(PATE). The PATE then predicts labels using differential privacy for another
unlabeled public data set. Since the resulting public dataset and its labels are
not private information, they can be used to train any other model, effectively
transferring knowledge from the PATE to a public model while preserving
privacy. Note that we do not require that the models used internally in creating
the PATE to be differentially private: only the aggregation step (used to
predict labels for the public dataset) involves differential privacy.

We split the training input dataset of MNIST digits into a bag of five parts.
Using the extension \lstinline|bmap| for each part, we independently train a
logistic regression model with non-private gradient descent. Assuming the input
dataset has sensitivity $1$, only one part of training data can change. \Fuzzi
correctly derives the fact that therefore at most $1$ trained logistic
regression model will change, resulting in a bag of model parameters with
sensitivity $1$.

We use the private ensemble of models to label another $100$ test images;
with privacy cost $\epsilon = 20.0$, and $\delta = 0.0$, we are able to reach
an accuracy of $0.82$. The large $\epsilon$ value here is related to the
small size
of the training set. To release a public label for a given image, the
private scores are collected from the PATE with \lstinline|bmap| and then a
noised average of the private scores is released by \lstinline|bsum| with the
Laplace mechanism. Since we only have $5$ private scores for each image, the
noise variance must be small so as not to destroy the utility of the scores,
resulting
in big $\epsilon$. To increase the stability of the released label (and hence decrease the privacy cost) we could increase the number of models, thus increasing the
number of private scores and the scale of the summed score, thereby allowing more noise added to their average. However, the result would be that each model would have been trained on correspondingly fewer images, resulting in worse classification performance on this dataset. On larger datasets, our \Fuzzi implementation of PATE
would provide the same level of classification performance with lower $\epsilon$
cost.
%

\subsection{Na\"ive Bayes}
\label{sec:eval-naive-bayes}
We next implement a simple spam detection algorithm using the Spambase dataset
from UCI Machine Learning Repository~\cite{Dua:2017}. The binary-labeled dataset
(spam or non-spam) consists of 57 features, mostly of word frequencies for a set
of given words, with additional features describing run lengths of certain
sequences. We binarize all features from the data set to simplify the
probability model described below (i.e., instead of how frequently a word
appears on a scale of $[0,100]$, we only know whether the word was used (1) or
not (0)). We can implement a more sophisticated Gaussian Na\"ve Bayes model that
takes advantage of the frequency data using the same principles as in this
experiment, but we chose to simplify the features to present a simpler model. We
use 4500 samples for training and 100 samples for evaluation. Our privacy goal
in this experiment is to limit an adversary's ability to guess whether a
particular document was used to train the classification model.

A key assumption of the Na\"ive Bayes model is that given the class $y$ of a data
point $\vec{x}$, all features are conditionally independent of each other. This
assumption allows us to decompose the joint probability $P(\vec{x}, y)$ into the
product $P(y)
\cdot \Pi_j
P(x_j | y)$, where $x_j$ represents the $j$-th coordinate of the binary vector
$\vec{x}$. In our experimental setup, the $j$-th coordinate of a data point
represents the presence of a word in the document. The goal of the Na\"ive Bayes
model is to estimate the probabilities of $P(y = 1)$ and $P(x_j = 1 | y = 1)$
and $P(x_j = 1 | y = 0)$ given the training data.
Thus, when we get a new document $\vec{x}'$, we can compare the
probabilities
$$P(y = 1) \Pi_j P(x_j = x_j' | y = 1) \leq?\; P(y = 0) \Pi_j P(x_j = x_j' | y =
0)$$ to make a prediction on whether the document is spam or not ($y=1$ or
$y=0$). Since we have $57$ features, this gives us $57 \cdot 2 + 1 = 115$ parameters
to estimate.

Estimating the parameter $P(y = 1)$ simply involves adding noise to the number
of spam documents in the training set and dividing that count by the noised
size of the training data set. We achieve this by first
applying \lstinline|bmap| to map each training data point to either $1$ or $0$
depending on its label, followed by \lstinline|bsum| and Laplace mechanism to
get a noised count. We can get a noised size of the training set through
applying Laplace mechanism the training set's size.

Estimating the parameters $P(x_j = 1 | y = 1)$ and $P(x_j = 1 | y = 0)$ follows
an essentially identical procedure. We first apply \lstinline|bmap| to each
training data point to map them to either $1$ or $0$ based on the values of
$x_j$ and $y$, followed by a \lstinline|bsum| operation to get the value of
these conditional counts. We already computed the noised count of training
samples with $y = 1$ when estimating $P(y = 1)$, but we perform the same
procedure to compute a noised count of $y = 0$. Dividing the noised conditional
counts by the noised counts of $y = 1$ and $y = 0$ gives us the final estimates
of the model parameters.

We achieve a training accuracy of $0.70$ and test accuracy of $0.69$, with
privacy costs $\epsilon = 7.70$ and $\delta = 0$. This classification accuracy
is only slightly worse than the accuracy $0.72$ of a non-private Na\"ive
Bayes model that we implemented using binarized features from the same dataset.

%
\subsection{K-Means}
\label{sec:eval-kmeans}
\hyphenation{un-supervised}
Finally, we perform a K-Means clustering experiment to evaluate \Fuzzi's
usability for an unsupervised learning task on the iris dataset from
\citet{fisher36lda}. This dataset contains three classes of iris flowers,
with 50 flowers from each class. Each flower comes with four numeric features of
petal and sepal dimensions and a label representing its class. Our experiment
randomly selected one data point from each of the three classes as the initial
public centroids; the \Fuzzi program uses \lstinline|partition| to map each
data point to its closest centroid and create partitions accordingly. Other
than the three data points used to initialize centroids, we used all other
data for
unsupervised training.  (This experiment assumes a small part---in this
case, three data points---of the training set is given as public
information. Past
work implementing differentially private K-Means made a similar
assumption~\cite{Reed:2010:DMT:1932681.1863568}.)

On each pass over the training set, we first compute a noised sum of data points
within each partition; we also compute noised sizes of each partition. We
use these values to compute each partition's average point as the new centroids
for the next pass. For evaluation, we classify all points within a partition
with the majority label, and we obtain the accuracy of the clustering with these
classifications. We do not use the labels for unsupervised training.

We found that the performance of the clustering algorithm varies depending on
the initial centroids selected: running the experiment 100 times, all within 5
passes over the data set, we reach lowest accuracy of 0.55 and highest accuracy
of 0.9, with a median accuracy of 0.69. Increasing the iteration count does not
reduce this spread. We implemented a non-private version of the same algorithm,
and achieved lowest accuracy $0.59$, highest accuracy $0.96$ and median accuracy
of $0.59$ on 100 experiments. Similar to Na\"ive Bayes, we see a slight drop in
classification accuracy compared to the non-private implementation.

Each run has privacy cost $\epsilon = 21.0$ and $\delta =
0.0$. The large $\epsilon$ cost here is again related to the small size of the
training set. In a small dataset, each data point has a larger impact on the
released centroids; in order to reach a reasonable level of classification
accuracy, we chose to apply the Laplace mechanism with a smaller noise level,
resulting in larger $\epsilon$ cost.



\section{Limitations}
\label{sec:limitations}

We briefly discuss some limitations and shortcomings of \Fuzzi.

\paragraph{Limitation of sensitivity}
\Fuzzi's type system interface strikes a careful balance between expressiveness and complexity.
Our approach is sufficient for expressing sensitivities of primitive values
such as $\tyint$ and $\tyfloat$ and can capture a top-level sensitivity for
vectors and bags; however, typing contexts cannot express sensitivities for
individual values within a vector. For example, McSherry and Mironov developed a
differentially private version of the recommender system based on Netflix user
ratings ~\cite{McSherry:2009:DPR:1557019.1557090}, where the sensitivity of
inputs to the system is defined by the changes that may happen to a single row
within a matrix, rather than the whole matrix. \Fuzzi currently cannot carry out
automatic inter-structure sensitivity derivation and cannot provide automatic
differential privacy checking for McSherry and Mironov's algorithm.

\paragraph{Lack of support for abstraction}
Vectors and bags are well studied objects in the differential privacy
literature, and they have first class support in \Fuzzi. However, \Fuzzi does
not provide facilities to specify general abstract data types and their
neighboring relations. \Fuzzi must know how to translate neighboring relation
into an \aprhl assertion, and this translation is not currently extensible. This
limitation may force programmers to contort their code in order to represent a
high-level concept through arrays. An example algorithm that cannot be
adequately expressed in \Fuzzi due to lack of abstraction is the {\em binary
mechanism}~\cite{Chan:2011:PCR:2043621.2043626}, which builds a tree of partial
sums of the input data and accumulates a statistic whose sensitivity is
proportional to the depth of the tree.

\paragraph{Potential Vulnerabilities}
\Fuzzi's semantics uses real numbers as a model for the type $\tyfloat$. However,
the implementation uses floating point numbers. As shown
by \citet{Mironov:2012:SLS:2382196.2382264}, using the Laplace mechanism in
this setting may result in
vulnerable distributions that can compromise the original sensitive data.
Although \Fuzzi guarantees co-termination over neighboring data, it is
vulnerable to timing channel
attacks~\cite{Haeberlen:2011:DPU:2028067.2028100}. A \Fuzzi program that uses
sensitive loop conditions may result in vastly different execution duration.
This side channel allows an attacker to distinguish runs with high confidence.
The first issue can be alleviated by a careful implementation of Laplace
mechanism that incorporates Mironov's mitigation strategy, while the second
issue is more fundamental---\Fuzzi's type system needs to approximately measure
the execution time, which we did not address in this work.

\paragraph{Performance concern due to copy assignments}
\Fuzzi uses copy assignments for arrays. We have worked with relatively small
datasets in the experiments, and the sizes of these arrays have not caused
severe performance problems in our experiments. However, today's machine
learning tasks typically operate on datasets that are many orders of magnitude
larger, and \Fuzzi likely cannot handle computations over these datasets
efficiently. To adapt \Fuzzi's theory for a semantics that allows sharing, we
need to create a new flavor of \aprhl that can reason about heaps. One potential
direction is to integrate separation
logic~\cite{Reynolds:2002:SLL:645683.664578} into \aprhl.

\section{Related Work}
\label{sec:related-works}
\paragraph{Query languages}
McSherry introduced Privacy Integrade Queries (PINQ) as an embedded query
language extension for the C\# programming
language~\cite{McSherry:2009:PIQ:1559845.1559850}. PINQ pioneered language-level
support for differential privacy by analyzing sensitivities for SQL-like
queries and releasing noised results of these queries using the Laplace
mechanism. \Fuzzi's \lstinline|partition| extension takes inspiration
from PINQ's partition operator, adapting it to an imperative
program that computes over arrays.

The FLEX framework~\cite{Johnson:2018:TPD:3187009.3177733} allows programmers to
run differentially private SQL queries against a private database. FLEX uses
an \textit{elastic sensitivity} technique to support SQL queries with joins
based on equality. \Fuzzi focuses on adding support of Differential Privacy to
a general-purpose imperative language; however, the theory around elastic
sensitivities could inspire future extensions to the \Fuzzi type system.

DJoin~\cite{Narayan:2012:DDP:2387880.2387895} runs SQL queries
over databases distributed over many machines. The distributed nature of the
data is not just a question of size, but may be due to the fact that
different databases may be owned by different organizations that do not wich
to share them; there simply is no single
way to get all the data in the same place. (For example, analysts may
want to correlate travel data with illness diagnosis data, with the former
provided by airline companies while the latter provided by hospitals.) \Fuzzi
does not address distributed computations: it runs on a single
machine and assumes data is already in the memory of this machine.
\ifpostreview
\bcp{Blah:}%
Running
typechecked differentially private code over separately curated private data is
another research direction that we can pursue for \Fuzzi.
\fi

\paragraph{Fuzz and related languages}
Fuzz is a higher-order functional programming language with a
sensitivity-tracking type system and differentially private primitives
~\cite{Reed:2010:DMT:1932681.1863568}. \Fuzzi's sensitivity type system is
inspired by Fuzz, but differs in that \Fuzzi separately tracks the
sensitivity of each value in the store (which may be change as the program
assigns to variables), while Fuzz tracks only function sensitivity.  Also,
Fuzz's type system is restricted to $(\epsilon, 0)$-differential privacy,
while \Fuzzi generalizes this to $(\epsilon, \delta)$-differential privacy.

DFuzz~\cite{Gaboardi:2013:LDT:2429069.2429113} extends Fuzz with linear indexed
types and dependent types, allowing
programmers to abstract types over sensitivity annotations. Compared to DFuzz,
both Fuzz and \Fuzzi only allow purely numeric values as sensitivity annotations
in types. This additional level of expressiveness admits programs whose
sensitivities and privacy costs scale with input sensitivities. Although \Fuzzi
does not allow such indexed types, the extension mechanism does allow language
developers to add typing rules quantified over unknown constants (such as
the loop
count in \textsc{Adv-Comp} for advanced composition); this provides another way
for programmers to write programs whose privacy costs scale with program
constants.

AdaptiveFuzz~\cite{Winograd-Cort:2017:FAD:3136534.3110254} extends
Fuzz by using staged computation and stream semantics to
implement a powerful composition mechanism called Privacy
Filters ~\cite{NIPS2016:6170}. These give programmers the freedom to
run future computations based on results released from earlier differentially
private computations. This allows, for example, programmers to stop a private
gradient descent loop as soon as accuracy reaches a desired threshold, rather
than fixing the number of iterations ahead of time. \Fuzzi implements advanced
composition for improved privacy cost aggregation, but privacy filters are
not yet formalized in either \aprhl or \Fuzzi.

Duet~\cite{duet} is a higher-order functional language that provides
$(\epsilon, \delta)$-differential privacy. Fuzz's original type system relies on
composition properties that break down when generalized to cases where $\delta >
0$. As an example, an $(\epsilon, 0)$-DP Fuzz function $f$ that takes a
$1$-sensitive dataset as input has the property that, when $f$ runs on a
$2$-sensitive dataset, the privacy costs scales accordingly to $(2\epsilon,
0)$-DP. However, if $f$ is a general $(\epsilon, \delta)$-DP computation on
$1$-sensitive datasets, it is \textit{not} true that running $f$ on a
$2$-sensitive dataset is $(2\epsilon, 2\delta)$-DP. Duet solves this problem by
separating its type system into two disjoint parts: one that keeps track of
sensitivities and allows scaling and another that keeps track of privacy costs
and disallows scaling. In \Fuzzi, we have a similar separation: the typing
contexts of a command sequence are strict pre-conditions and post-conditions and
do not allow scaling, except for commands typechecked with the \texttt{linear}
typing judgements that explicitly allow scaling of sensitivities in the pre- and
post-condition. These \texttt{linear} commands are deterministic and cannot use
the Laplace mechanism by definition, so scaling their typing judgements'
sensitivities does not bring up the same issues that Fuzz has.
\paragraph{Verification systems}
\citet{Albarghouthi:2017:SCP:3177123.3158146} developed an automated differential
privacy proof synthesis system based on the idea of \textit{coupling}. \Fuzzi
and \aprhl use the same mathematical device  to simplify
relational reasoning between non-deterministic outputs from the Laplace
mechanism.

LightDP is an imperative language with semi-automatic typechecking that uses
dependent types to prove differential privacy properties of a program
~\cite{Zhang:2017:LTA:3093333.3009884}. LightDP's type system also keeps track
of distances of variables between executions on neighboring inputs. A major
difference from
\Fuzzi is
that LightDP's type system tracks the {\em exact} distance between variables
through a dependent type system, while \Fuzzi tracks upper bounds on the
distance between variables. LightDP elaborates the source code into a slightly
extended language that explicitly keeps track of privacy costs in a
distinguished variable and then uses a MaxSMT solver to discharge the generated
verification conditions in the process of typechecking. \Fuzzi typing rules'
soundness are proven ahead of time, and \Fuzzi's sensitivity checking process
does not generate new proof obligations. Due to this design, \Fuzzi does not
require a constraint solver to aid in typechecking.

The EasyCrypt toolset allows developers to construct machine-checkable proofs of
relational properties for probabilistic
computations~\cite{EasyCrypt:manual}. EasyCrypt has a development branch that
focuses on Differential Privacy verification through \aprhl. EasyCrypt provides
built-in support of \aprhl proof rules, and also supports termination analysis
through a compatible program logic pHL (Probabilistic Hoare Logic). Fuzzi's
development does not connect with EasyCrypt's \aprhl implementation, but this is
a potential future direction for rigorously checking Fuzzi's theories.
\paragraph{Testing Differential Privacy}
\citet{Ding:2018:DVD:3243734.3243818} developed a statistical testing framework for
detecting violations of differential privacy of Python programs. This framework
performs static analysis on Python code and generates inputs that seem
likely to violate differential privacy based on this analysis. It also repeatedly executes the program to collect statistical evidence of
violations of differential privacy. This framework demonstrates the
potential for a lighter-weight approach to providing differential privacy
guarantees; we
could potentially apply the same methodology to aid \Fuzzi extension
designers by testing
typing rules before formally proving their soundness.

\paragraph{High-level frameworks}
The PSI private data sharing interface~\cite{DBLP:journals/corr/GaboardiHKNUV16}
is designed to enable non-expert users and researchers to both safely deposit
private data and apply a select set of differentially private algorithms to
collect statistics from the deposited data. \Fuzzi, on the other hand, is
designed only for the task of implementing differentially private
algorithms. It
expects its users to have some familiarity with key concepts such as sensitivity
and privacy budget, and it allows power users to extend its typechecker for more
sophisticated programs.

The $\epsilon$KTELLO Framework~\cite{Zhang2018EkteloAF} provides a set of
expressive  high-level combinators for composing algorithms, with the guarantee
that any algorithm composed of $\epsilon$KTELLO combinators automatically
satisfies differential privacy. The $\epsilon$KTELLO framework allows users
to customize differentially private algorithms in order to achieve higher
utility from querying private data. \Fuzzi, by contrast, is an attempt at
building a rather
{\em low-level} core language with support for Differential Privacy.  A future
direction could be to build a high-level framework like
$\epsilon$KTELLO
over \Fuzzi, providing both expressive combinators and automatic
verification of Differential Privacy for the implementation of these combinators
in the same system.

\section{Conclusion and Future Work}
\label{sec:future}

The rise of Differential Privacy calls for reliable, yet familiar tools that
help programmers control privacy risks. \Fuzzi gives programmers a standard
imperative language with automated privacy checking, which can be enriched
by expert users with extensions whose privacy properties are proved in a
special-purpose relational logic.

Many avenues for improvement still remain.
\begin{enumerate*}
\item We can enrich the set of \Fuzzi extensions to further increase
\Fuzzi's utility. For example, adding Report Noisy Max would allow an analyst to
find the largest value in a vector with small privacy cost; the Exponential
mechanism would allow programs to release categorical data (as opposed to
numerical) with differential privacy
guarantees~\cite{Dwork:2014:AFD:2693052.2693053}. We expect both mechanisms
can be formalized in \aprhl~\cite{Barthe:2016:PDP:2933575.2934554} and added to
\Fuzzi.
\item We can engineer typechecker plugins that dynamically load new
extension typing rules to make \Fuzzi's implementation more flexible.  (At
the moment, adding an extension typing rule requires editing the typechecker sources.)
\item We can formalize Privacy Filters~\cite{NIPS2016:6170} in \aprhl and add adaptive composition to \Fuzzi. This would allow programmers to use the adaptive aggregation mechanism from AdaptiveFuzz in \Fuzzi.
\item We can implement \Fuzzi as a formalized framework in
Coq~\cite{Coq:manual}. This would allow power users to write
machine-checked proofs of extension typing rules.
\item We can incorporate proof synthesis techniques to automatically search
for privacy proofs for extensions, following
\citet{Albarghouthi:2017:SCP:3177123.3158146}, who demonstrated the
effectiveness of synthesizing privacy proofs for interesting differential
privacy mechanisms.  Proof synthesis could streamline
prototyping new \Fuzzi extensions and their typing rules.
\end{enumerate*}

\section{Acknowledgments}
\label{sec:acks}

We are grateful to Justin Hsu, David Darais, and Penn PLClub for their
comments, and we thank the anonymous ICFP reviewers for their detailed and
helpful feedback. This work was supported in part by the National Science
Foundation under grants CNS-1065060 and CNS-1513694.

\bibliography{local}

\appendix

\ifextended

\section{Semantics}
\label{ap:semantics}
\begin{figure}[t]
\begin{align*}
\mathtt{bind} &:: \distr\, A \rightarrow (A \rightarrow \distr\, B) \rightarrow \distr\, B\\
\mathtt{ret} &:: A \rightarrow \distr A
\end{align*}
\caption{Operations over the distribution monad}
\label{fig:distr-monad-op}
\end{figure}
\Fuzzi's semantics directly follows from the work of \citet{Barthe:2016:PDP:2933575.2934554},
but we extend the language with operations over vectors and bags
(\Cref{fig:cmd-sem}). Due to the possibility of out-of-bounds indexing, \Fuzzi's
semantics for expression accounts for partiality by modeling expressions with
partial functions.
%
%
We use a list structure to model arrays in \Fuzzi.
The function $\mathtt{length}\; v$ returns the length of the list $v$. The
function $\mathtt{resize}\; \mathit{len}\; v$ updates the list $v$ such that it
has length equal to $\mathit{len}$ and pads the list with well-shaped default
values if necessary. The function $\mathtt{update}\; b\; i\; v$ returns another
list whose $i$-th element is set to $b$, or $\bot$ if $i$ is out of bounds. We
elide an implicit coercion from $\bot$ to the special distribution
$\mathtt{distr_0}$. The distribution $\mathtt{distr_0}$ is the empty
distribution; it is used to model non-termination.

The semantics of \Fuzzi commands are given as probabilistic functions from
program states to sub-distributions over program states. We write $\bigcirc\, A$
for the set of sub-distributions over $A$. Sub-distributions forms a monad with
two operators $\mathtt{ret}$ and $\mathtt{bind}$
(\Cref{fig:distr-monad-op}). The $\mathtt{ret}$ operator takes a value and
produces a distribution whose entire mass is concentrated on that single value;
the $\mathtt{bind}$ operator builds conditional distributions through
compositions between distributions and functions from samples to distributions.
\begin{figure}[t]
\begin{small}
\begin{tabular}{L}
\denote{x = e} = \\
\hspace{1em}\lambda M. \llet\, \mathtt{v} = \denote{e}\, M\, \lin\, \mathtt{ret}(M[x\mapsto v])\\
\denote{x[i] = e} = \\
\hspace{1em}\lambda M. \llet\, \mathtt{v} = \denote{e}\, M\, \lin\, \\
\hspace{2.6em}\llet\, \mathtt{idx} = \denote{i}\, M\,\lin\, \\
\hspace{2.6em}\mathtt{ret}(M[x \mapsto \mathtt{update}\, \mathtt{v}\, \mathtt{idx}\, M(x)]) \\
\denote{x.length = e} =\\
\hspace{1em}\lambda M. \llet\, \mathtt{len} = \denote{e}\, M\, \lin\, \\
           \hspace{2.6em}\mathtt{ret}(M[x \mapsto \mathtt{resize}\, \mathtt{len}\, M(x) ]) \\
\denote{\cmdif{e}{c_1}{c_2}} = \\
\hspace{1em}\lambda M. \llet\, \mathtt{cond} = \denote{e}\, M\, \lin\, \\
           \hspace{2.6em}\mathtt{if}\, \mathtt{cond}\,
           \mathtt{then}\, \denote{c_1}\, M\,
           \mathtt{else}\, \denote{c_2}\, M \\
\denote{\cmdwhile{e}{c}}_{n} =\\
\hspace{1em}\lambda\, M. \mathtt{if}\, n \leq 0\\
\hspace{2.7em}\mathtt{then}\, \mathtt{distr}_0\\
\hspace{2.7em}\mathtt{else}\, \mathtt{if}\, \denote{e}\, M\\
\hspace{2.7em}\mathtt{then}\, \mathtt{bind}\,(\denote{c}\,M)\,(\denote{\cmdwhile{e}{c}}_{n-1})\\
\hspace{2.7em}\mathtt{else}\, \mathtt{ret}\, M\\
\denote{\cmdwhile{e}{c}} = \bigsqcup_n\,\denote{\cmdwhile{e}{c}}_n
\end{tabular}
\end{small}
\caption{Semantics of \Fuzzi}
\label{fig:cmd-sem}
\end{figure}

\section{Proof Rules for \aprhl}
\label{ap:aprhl}
Reasoning about programs that implement differentially private mechanisms
involves both relational and probablistic reasoning. The authors of \aprhl
combined ingredients from program logic and probability theory to create a proof
system that manages the complexity of such proofs. The union of two key
ingredients--- relational Hoare logic and approximate liftings---results in a
proof system that is very well suited for differential privacy proofs.

\citet{Benton:2004:SRC:964001.964003} established relational Hoare logic (RHL).
RHL provides a system for writing down proofs of properties between two
executions of deterministic while-programs. This proof system would allow us to
embed a type system that tracked sensitivity (but not privacy costs) for the
deterministic fragment of \Fuzzi; however, RHL lacks the reasoning facility for
distributions.

Distributions over program states have complicated structures that arise from
branches and loops; proofs that directly manipulated these distributions may
get unwieldy due to the complexity of these distributions'
structures. Fortunately, \aprhl applies ``approximate liftings'' to
significantly simplify such proofs over distributions.

An approximate lifting is characterized by a relation $R$ between two
distributions' support: given $\mu_A$ a distribution in $\distr\,A$ and $\mu_B$
in $\distr\, B$, the relation $R$ is a subset of $A \times B$. This approximate
lifting based on $R$ allows us to only consider elements linked by $R$ while
proving relational properties on $\mu_A$ and $\mu_B$, and \aprhl applies this
abstraction to simplify proofs of probabilistic programs.

An \aprhl judgment has the form
$\aprhlstmt{c_1}{(\epsilon, \delta)}{c_2}{\Phi}{\Psi}$. The assertions $\Phi$
and $\Psi$ are relations over program states. The pre-condition $\Phi$ is a
deterministic relation that the input program states to $c_1$ and $c_2$
satisfies. Since $c_1$ and $c_2$ are probabilistic programs, the output program
states after running $c_1$ and $c_2$ are distributions of program states, and
the post-condition is used to construct an approximate lifting between the two
output program state distributions. The validity of an \aprhl judgment implies
a valid approximate lifting of the post-condition $\Psi$ over $\denote{c_1} M_1$
and $\denote{c_2} M_2$. The relations $\Phi$ and $\Psi$ may describe properties
as $x\ag{1} = x\ag{2}$---the value held in $x$ is the same in both executions.

An approximate lifting of $R \subseteq A_1 \times A_2$ that relates two
distributions $\mu_1 : \distr\, A_1$ and $\mu_2 : \distr\, A_2$ is justified by
two witness distributions and two cost parameters $(\epsilon, \delta)$. Let
$\star$ be a distinct element not in $A_1 \cup A_2$, and write the extended
relation $R^\star = R \cup A_1 \times \{\star\} \cup \{\star\} \times A_2$. Let
$A_1^\star$ and $A_2^\star$ be $A_1$ and $A_2$ extended with $\star$,
respectively. We define the $\epsilon$-distance between $\mu_L$ and $\mu_R$ as
following:
$$d_\epsilon(\mu_L, \mu_R) = \max\limits_{S \subseteq A_1^\star \times
A_2^\star}(\mu_L(S) - \exp(\epsilon)\cdot\mu_R(S))$$ Then the two witness
distributions $\mu_L$ and $\mu_R$ should satisfy the following
properties~\cite{DBLP:journals/corr/abs-1710-09951}:

\begin{tabular}{@{}rlr@{}}
1. & $\pi_1(\mu_L) = \mu_1$ and $\pi_2(\mu_R) = \mu_2$ & \hfill \textsc{Marginal}\\
2. & $\mathtt{supp}(\mu_L) \cup \mathtt{supp}(\mu_R) \subseteq R^\star$ & \hfill \textsc{Support}\\
3. & $d_\epsilon(\mu_L, \mu_R) \leq \delta$ & \hfill \textsc{Distance}
\end{tabular}

In the \textsc{Marginal} condition, $\pi_1$ and $\pi_2$ are first and second
projections of distributions that produces the first and second marginal
distributions on $A_1^\star \times A_2^\star$, this condition requires the
witnesses to ``cover'' the original $\mu_1$ and $\mu_2$ on their respective
projections.

The \textsc{Support} condition requires support of both witnesses to reside
within the relation $R^\star$.

The \textsc{Distance} condition gives bounds on the $\exp(\epsilon)$ multiplicative
difference and the $\delta$ additive difference between these two witness
distributions, and this distance definition matches up with the definition of
privacy costs used in differential privacy.

We can gain some intuition of approximate liftings
from~\Cref{fig:coupling}. Here we visualize an approximate lifting of the
equality relation over two Laplace distributions. The validity of the
approximate lifting gives a global bound on the difference in probability of the
linked elements in this plot, and this bound can be naturally interpreted as
privacy costs for differential privacy.
\begin{figure}[t]
\includegraphics[width=0.5\textwidth]{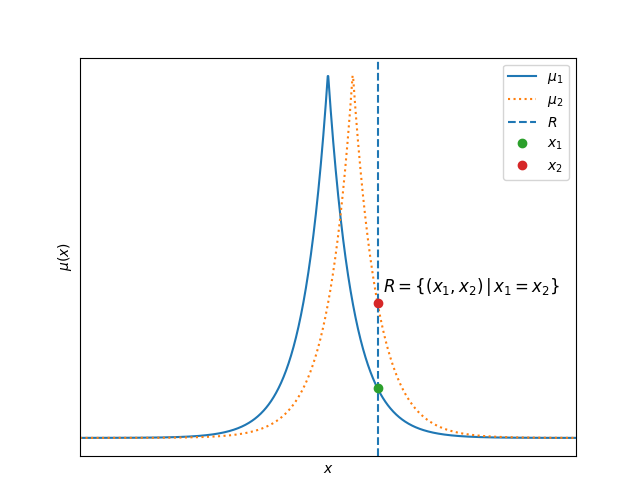}
\caption{Approximate lifting of $R$ between two Laplace distributions}
\label{fig:coupling}
\end{figure}

The proof rules from \aprhl constructs these witness distributions for
probabilistic imperative programs. In particular, the \textsc{Seq} rule allows
us to treat the approximate lifting of a post-condition as the pre-condition of
the next command. The resulting proof system effectively abstracts away explicit
reasoning of distributions. The relational assertions of \aprhl allows us to
naturally express predicates over pairs of program states, and approximate
liftings cleverly hides much of the plumbing for probabilistic reasoning. We have
listed a subset of \aprhl proof rules used in the development of \Fuzzi
in \Cref{apRHL:part1} and \Cref{apRHL:part2}.

The structural \aprhl proof rules relate programs that are structurally similar:
assignments are related with assignments, conditionals are related with
conditionals by relating their corresponding true and false branches, and loops
are related by synchronizing their loop bodies. Among the structural rules,
the \textsc{Lap} rule formalizes the Laplace mechanism---the pre-condition
states that a value to be released has sensitivity $k$, and releasing this
private value with noise scaled with $1/\epsilon$ is $(k\epsilon, 0)$-DP.

An important quirk of \aprhl is that it does \emph{not} have a conjunction rule
like this:
\begin{mathparsmall}
\inferrule[]
          {\aprhlstmt{c_1}{(\epsilon, \delta)}{c_2}{\Phi}{\Psi} \\\\
           \aprhlstmt{c_1}{(\epsilon, \delta)}{c_2}{\Phi}{\Theta}}
          {\aprhlstmt{c_1}{(\epsilon, \delta)}{c_2}{\Phi}{\Psi \land \Theta}}
\end{mathparsmall}%
Since an \aprhl judgement is justified by the existence of witness distributions
for the approximate lifting of the post-condition, the two judgements above the
inference bar tells us there exists witnesses separately justifying the
approximate liftings of $\Psi$ and $\Theta$. However, this does not guarantee
the existence of an approximate lifting for their conjunction. Thus, this rule
is not valid.
However, \aprhl does have a \textsc{Frame} rule, which allows us to conjunct
$\Theta$ with the pre- and post-conditions, as long as $\Theta$ does not mention
any modified variables of the two related programs.

Sometimes structural \aprhl rules are not enough, since sensitive values may
steer program control flow towards different sequences of code. The authors
of \aprhl account for reasoning of these programs through one-sided proof rules,
as shown in \Cref{apRHL:part2}.

The rules \textsc{While-L} and \textsc{While-R} deserve some special attention:
it allows us to relate a while loop with the skip command using a one-sided loop
invariant. This allows us to carry out standard Floyd-Hoare logic reasoning on a
single while loop. We will use these one-sided loop rules extensively later in
the proof of extension typing rules. These one-sided loop rules require a
side-condition of {lossless}-ness for the loop body---a program $c$ is lossless
if executing $c$ results in a proper distribution, i.e., the probability
distribution sums up to $1$. Although \aprhl gives the definition of lossless,
it does not provide proof rules for lossless-ness. In the development of \Fuzzi,
we developed a termination type system compatible with this definition of
lossless. Details can be found in \Cref{ap:aux-checker}.

\begin{figure*}[t]
\begin{mathparsmall}
\inferrule[Skip]
          { } {\aprhlstmt{\cmdskip}{(0, 0)}{\cmdskip}{\Phi}{\Phi}}

\inferrule[Assn]
          { }
          {\aprhlstmt
           {x_1\ag{1} = e_1\ag{1}}
           {(0, 0)}
           {x_2\ag{2} = e_2\ag{2}}
           {\Phi[e_1\ag{1}, e_2\ag{2}/x_1\ag{1}, x_2\ag{2}]}{\Phi}}

\inferrule[Lap]
          {\Phi \triangleq |e_1\ag{1} - e_2\ag{2}| \leq k}
          {\aprhlstmt
                {x_1\ag{1} = \laplace{1/\epsilon}{e_1\ag{1}}}
                {(k\epsilon, 0)}
                {x_2\ag{2} = \laplace{1/\epsilon}{e_2\ag{2}}}
                {\Phi}
                {x\ag{1} = x\ag{2}}
          }

\inferrule[Seq]
          {\aprhlstmt{c_1}{(\epsilon, \delta)}{c_2}{\Phi}{\Psi} \quad
           \aprhlstmt{c_1'}{(\epsilon', \delta')}{c_2'}{\Psi}{\Theta}}
          {\aprhlstmt{c_1;c_1'}{(\epsilon + \epsilon', \delta + \delta')}{c_2;c_2'}{\Phi}{\Theta}}

\inferrule[Cond]
          {\vDash \Phi \Rightarrow e_1\ag{1} = e_2\ag{2} \\\\
           \aprhlstmt{c_1}{(\epsilon, \delta)}{c_2}{\Phi \land e_1\ag{1}}{\Psi} \\\\
           \aprhlstmt{c_1'}{(\epsilon, \delta)}{c_2'}{\Phi \land \neg e_1\ag{1}}{\Psi}}
          {\aprhlstmt{\cmdif{e_1}{c_1}{c_1'}}
                     {(\epsilon, \delta)}{\cmdif{e_2}{c_2}{c_2'} }{\Phi}{\Psi}}
%

\inferrule[While*]
          {\vDash \Phi \Rightarrow e_1\ag{1} = e_2\ag{2} \\\\
           \aprhlstmt{c_1}{(0, 0)}{c_2}{\Phi\land e_1\ag{1}}{\Phi}
          }
          {\aprhlstmt{\cmdwhile{e_1}{c_1}}{(0, 0)}{\cmdwhile{e_2}{c_2}}{\Phi}{\Phi \land \neg e_1\ag{1}}}

\inferrule[Conseq]
          {\aprhlstmt{c_1}{(\epsilon', \delta')}{c_2}{\Phi'}{\Psi'} \\\\
           \vDash \Phi \Rightarrow \Phi' \quad
           \vDash \Psi' \Rightarrow \Psi \quad
           \vDash \epsilon' \leq \epsilon \quad
           \vDash \delta' \leq \delta}
          {\aprhlstmt{c_1}{(\epsilon, \delta)}{c_2}{\Phi}{\Psi}}

\inferrule[Equiv]
          {\aprhlstmt{c_1'}{(\epsilon, \delta)}{c_2'}{\Phi}{\Psi} \quad
           c_1 \equiv c_1' \quad
           c_2 \equiv c_2'}
          {\aprhlstmt{c_1}{(\epsilon, \delta)}{c_2}{\Phi}{\Psi}}

\inferrule[Frame]
          {\aprhlstmt{c_1}{(\epsilon, \delta)}{c_2}{\Phi}{\Psi}
           \quad \fv\, \Theta \cap \mv\, (c_1, c_2) = \varnothing
          }
          {\aprhlstmt{c_1}{(\epsilon, \delta)}{c_2}{\Phi \land \Theta}{\Psi \land \Theta}}

\end{mathparsmall}
\caption{Proof rules for \aprhl (Structural Rules)}
\label{apRHL:part1}
\end{figure*}

\begin{figure*}[t]
\begin{mathparsmall}
\inferrule[While-L]
          {\aprhlstmt{c_1}{(0, 0)}{\cmdskip}{\Phi \land e_1\ag{1}}{\Phi} \\\\
           \vDash \Phi \Rightarrow \Phi_1\ag{1} \quad
           \Phi_1\ag{1} \vDash \cmdwhile{e_1}{c_1}\,\,\text{lossless}}
          {\aprhlstmt{\cmdwhile{e_1}{c_1}}{(0, 0)}{\cmdskip}{\Phi}{\Phi\land \neg e_1\ag{1}}}

\inferrule[While-R]
          {\aprhlstmt{\cmdskip}{(0, 0)}{c_2}{\Phi \land e_2\ag{2}}{\Phi} \\\\
           \vDash \Phi \Rightarrow \Phi_2\ag{2} \quad
           \Phi_2\ag{2} \vDash \cmdwhile{e_2}{c_2}\,\,\text{lossless}}
          {\aprhlstmt{\cmdskip}{(0, 0)}{\cmdwhile{e_2}{c_2}}{\Phi}{\Phi\land \neg e_2\ag{2}}}

\inferrule[Assn-L]
          { }
          {\aprhlstmt{x_1 = e_1}{(0, 0)}{\cmdskip}{\Phi[e_1\ag{1}/x_1\ag{1}]}{\Phi}}

\inferrule[Assn-R]
          { }
          {\aprhlstmt{\cmdskip}{(0, 0)}{x_2 = e_2}{\Phi[e_2\ag{2}/x_2\ag{2}]}{\Phi}}

\inferrule[Cond-L]
          {\aprhlstmt{c_1}{(\epsilon, \delta)}{c}{\Phi\land e_1\ag{1}}{\Psi}
           \qquad \aprhlstmt{c_1'}{(\epsilon, \delta)}{c}{\Phi\land \neg e_1\ag{1}}{\Psi}}
          {\aprhlstmt{\cmdif{e_1}{c_1}{c_1'}}{(\epsilon, \delta)}{c}{\Phi}{\Psi}}

\inferrule[Cond-R]
          {\aprhlstmt{c}{(\epsilon, \delta)}{c_2}{\Phi\land e_2\ag{2}}{\Psi}
           \qquad \aprhlstmt{c}{(\epsilon, \delta)}{c_2'}{\Phi\land \neg e_2\ag{2}}{\Psi}}
          {\aprhlstmt{c}{(\epsilon, \delta)}{\cmdif{e_2}{c_2}{c_2'}}{\Phi}{\Psi}}
\end{mathparsmall}
\caption{Proof rules for apRHL (One-sided Rules)}
\label{apRHL:part2}
\end{figure*}

\section{Rules for $\mathtt{term}$ and $\mathtt{linear}$}
\label{ap:aux-checker}
\begin{figure}[t]
\begin{mathparsmall}
\inferrule{x \in \Gamma}{\term{\Gamma}{x}}

\inferrule{\mathit{lit} \in \tyint \lor \mathit{lit} \in \tyfloat \lor \mathit{lit} \in \tybool}{\term{\Gamma}{\mathit{lit}}}

\inferrule{\term{\Gamma}{e_1} \qquad \term{\Gamma}{e_2}}{\term{\Gamma}{e_1\,\mathtt{op}\,e_2}}

\inferrule{\term{\Gamma}{e}}{\term{\Gamma}{e.\mathit{length}}}

\inferrule{ }{\term{\Gamma}{\cmdskip}}

\inferrule{\term{\Gamma}{e}}{\term{\Gamma}{x = e}}

\inferrule{\term{\Gamma}{c_1}\qquad \term{\Gamma}{c_2}}{\term{\Gamma}{c_1; c_2}}

\inferrule{\term{\Gamma}{e} \qquad \term{\Gamma}{c_1}\qquad \term{\Gamma}{c_2}}{\term{\Gamma}{\cmdif{e}{c_1}{c_2}}}

\inferrule{\term{\Gamma}{c}}
          {\term{\Gamma}{\cmdbmap{\mathit{in}}{\mathit{out}}{\mathit{t_{in}}}{i}{\mathit{t_{out}}}{c}}}

\inferrule{\term{\Gamma}{c}}
          {\term{\Gamma}{\cmdamap{\mathit{in}}{\mathit{out}}{\mathit{t_{in}}}{i}{\mathit{t_{out}}}{c}}}

\inferrule{\term{\Gamma}{c}
           \qquad \term{\Gamma}{\mathit{nParts}}}
          {\term{\Gamma}
                {\bm{\mathtt{partition}}
                (\mathit{in},\mathit{out},\mathit{t_{in}},i,\mathit{t_{out}},\mathit{t_{idx}},\mathit{out_{idx}},\mathit{t_{part}},\mathit{nParts},c)}
          }

\inferrule{ }{\term{\Gamma}{\bm{\mathtt{bsum}}(\mathit{in},\mathit{out},i,\mathit{t_{in}},\mathit{bound})}}
\end{mathparsmall}
\caption{Rules for $\mathtt{term}$}
\end{figure}

\begin{figure}[t]
\begin{mathparsmall}
\inferrule{ }{\tytriple{\Gamma}{\cmdskip}{\Gamma, (0, 0)}\,\linear}

\inferrule{\hastype{\Gamma}{e}{\sigma}{\tau}}
          {\tytriple{\Gamma}{x = e}{\Gamma[x \mapsto \sigma], (0, 0)}\,\linear}

\inferrule{\tytriple{\Gamma_1}{c_1}{\Gamma_2, (0, 0)}\,\linear
           \qquad \tytriple{\Gamma_2}{c_2}{\Gamma_3, (0, 0)}\,\linear}
          {\tytriple{\Gamma_1}{c_1;c_2}{\Gamma_3, (0, 0)}\,\linear}

\inferrule{\tytriple{\Gamma}{c_1}{\Gamma_1, (0, 0)}\,\linear
           \qquad \tytriple{\Gamma}{c_2}{\Gamma_2, (0, 0)}\,\linear
           \qquad \hastype{\Gamma}{e}{0}{\tybool}}
          {\tytriple{\Gamma}{\cmdif{e}{c_1}{c_2}}{\mathtt{max}(\Gamma_1, \Gamma_2), (0, 0)}\,\linear}

\mprset{flushleft}
\inferrule{\tinytext{Termination, Deterministic}
           \\\\ \term{\Gamma}{c}
           \qquad \determ\; c
           \\\\ \tinytext{Should Not Modify}
           \\\\ \mathit{t_{in}}, \mathit{in},\mathit{out},\mathit{i} \notin \mv\; c
           \\\\ \tinytext{Abbreviation}
           \\\\ \sigma\hspace{0.3em} = [\mv\, c, i, in, out \mapsto \infty]
           \\\\ \sigma' = [\mv\,c, i, \mathit{t_{in}}, \mathit{t_{out}} \mapsto \infty]
           \\\\ \tinytext{Dependency}
           \\\\ \tytriple
           {\mathtt{stretch}\;\Gamma[\mathit{t_{in}} \mapsto \hspace{0.38em}0]\sigma}
           {c}
           {\Gamma_1, (0, 0)}
           \\\\ \tytriple{\mathtt{stretch}\;\Gamma[\mathit{t_{in} \mapsto \infty}]\sigma}
                 {c}
                 {\Gamma_2, (0, 0)}
           \\\\ \Gamma_1(t_\mathit{out}) = 0
           \\\\ \{x \mathrel{|} x\in\mv\; c \land \Gamma_2(x) > 0\} \subseteq \{t_{\mathit{out}}\}
           \\\\ \tinytext{Output Sensitivity}
           \\\\ \Gamma_1(\mathit{t_{out}}) = 0
           \quad \Gamma_{\mathit{out}} =
           \Gamma
           [\mathit{out} \mapsto \Gamma(\mathit{in})]\sigma'
           }
          {\tytriple
            {\Gamma}
            {\cmdbmap{\mathit{in}}{\mathit{out}}%
                     {t_\mathit{in}}{i}{t_\mathit{out}}{c}}
            {\Gamma_{\mathit{out}}, (0, 0)}\,\linear}

\inferrule{\tinytext{Termination, Deterministic}
           \\\\ \term{\Gamma}{c}
           \qquad \determ\, c
           \\\\ \tinytext{Should Not Modify}
           \\\\ \mathit{t_{in}}, \mathit{in},\mathit{out},\mathit{i} \notin \mv\, c
           \\\\ \tinytext{Abbreviation}
           \\\\ \sigma\hspace{0.3em} = [\mv\, c, i, in, out \mapsto \infty]
           \\\\ \sigma' = [\mv\, c, i, \mathit{t_{in}}, \mathit{t_{out}} \mapsto \infty]
           \\\\ \tinytext{Dependency}
           \\\\ \tytriple
           {\mathtt{stretch}\,\Gamma[\mathit{t_{in}} \mapsto \hspace{0.38em}0]\sigma}
           {c}
           {\Gamma_1, (0, 0)}
           \\\\ \tytriple{\mathtt{stretch}\,\Gamma[\mathit{t_{in}} \mapsto \infty]\sigma}
                 {c}
                 {\Gamma_2, (0, 0)}
           \\\\ \Gamma_1(t_\mathit{out}) = 0
           \\\\ \{x | x\in\mv\, c \land \Gamma_2(x) > 0\} \subseteq \{t_{\mathit{out}}\}
           \\\\ \tinytext{Linear}
           \\\\ \tytriple{\Gamma[\mathit{t_{in}}\mapsto 1]\sigma}
                 {c}
                 {\Gamma_3, (0, 0)}\,\linear
           \\\\ \tinytext{Output Sensitivity}
           \\\\ \Gamma_{\mathit{out}} =
           \Gamma[\mathit{out} \mapsto \Gamma(in) \cdot \Gamma_3(\mathit{t_{out}})]\sigma'
           }
          {\tytriple
            {\Gamma}
            {\cmdamap{\mathit{in}}{\mathit{out}}{t_\mathit{in}}{i}{t_\mathit{out}}{c}}
            {\Gamma_{\mathit{out}}, (0, 0)}\,\linear}
\end{mathparsmall}
\caption{Rules for $\linear$ (Part 1)}
\end{figure}

\begin{figure}[t]
\begin{mathparsmall}
\mprset{flushleft}
\inferrule{
           \tinytext{Termination, Deterministic}
           \\\\ \term{\Gamma}{c} \quad \determ\, c
           \\\\ \tinytext{Should Not Modify}
           \\\\ \mathit{t_{in}}, \mathit{in}, \mathit{out}, i, \mathit{out_{idx}} \notin \mv\, c
           \\\\ \tinytext{Abbreviation}
           \\\\ \sigma\hspace{0.3em} = [\mv\, c, i, in, \mathit{out_{idx}} \mapsto \infty]
           \\\\ \sigma' = [\mv\, c, i,
                           \mathit{t_{in}}, \mathit{t_{idx}}, \mathit{out_{idx}},
                           \mathit{t_{part}} \mapsto \infty]
           \\\\ \tinytext{Number of Partitions Non-Sensitive}
           \\\\ \hastype{\Gamma}{\mathit{nParts}}{0}{\tyint}
           \\\\ \fv\, \mathit{nParts} \cap \mv\, \mathtt{c} = \varnothing
           \\\\ i, \mathit{t_{in}}, \mathit{t_{idx}}, \mathit{out_{idx}}, \mathit{t_{part}} \notin \fv\,\mathit{nParts}
           \\\\ \tinytext{Dependency}
           \\\\
           \tytriple
           {\mathtt{stretch}\,\Gamma[\mathit{t_{in}} \mapsto \hspace{0.38em}0]\sigma}
           {c}
           {\Gamma_1, 0}
           \\\\ \tytriple{\mathtt{stretch}\,\Gamma[\mathit{t_{in}} \mapsto \infty]\sigma}
                 {c}
                 {\Gamma_2, 0}
           \\\\ \Gamma_1(\mathit{t_{out}}) = 0
           \\\\ \{x | x\in\mv\, c \land \Gamma_2(x) > 0\} \subseteq \{t_{\mathit{out}}\}
           \\\\ \tinytext{Output Sensitivity}
           \\\\ \Gamma_{\mathit{out}} =
           \Gamma[\mathit{out} \mapsto \Gamma(\mathit{in})]
                 \sigma'
          }
          {\tytriple
            {\Gamma}
            {\bm{\mathtt{partition}}
              (\mathit{in}, \mathit{out},
              t_\mathit{in}, i, \mathit{t_{out}}, \mathit{t_{idx}}, \mathit{out_{idx}},
              \mathit{t_{part}}, \mathit{nParts}, c)}
            {\Gamma_{\mathit{out}}, (0, 0)}\,\linear}

\mprset{center}
\inferrule{\literal\,\mathit{bound} \\\\ \mathit{bound} \geq 0\\\\
           \phi = \Gamma(\mathit{in}) \quad \Gamma_\mathit{out}
           = \Gamma[\mathit{out} \mapsto \phi\cdot \mathit{bound}][i, \mathit{t_{in}} \mapsto \infty]
          }
          {\tytriple
            {\Gamma}
            {\bm{\mathtt{bsum}} (\mathit{in}, \mathit{out}, i, t_\mathit{in}, \mathit{bound})}
            {\Gamma_\mathit{out}, (0, 0)}
          }
\end{mathparsmall}
\caption{Rules for $\linear$ (Part 2)}
\end{figure}

\subsection{Justifying $\mathtt{term}$ rules}
In \Cref{sec:extension}, we defined two auxiliary properties of \Fuzzi
programs---$\mathtt{term}$ and $\mathtt{linear}$---in order to typecheck various
extensions. The typing rules for these auxillary properties share a similar
design as the typing rules for sensitivity---we give the definition of each
auxillary property in the base logic $\mathcal{L}$ and give the typing rules as
theorems to be justified in $\mathcal{L}$. However, for expressions, we again
use inductive relations instead of foundational definition since we do not plan
on extending the rules for expressions.

\begin{defn}
\label{defn:well-shaped}
The well-shaped judgment $M \in \shape(\Gamma)$ for memories is defined as:
for any $x \in_\sigma \tau$ in $\Gamma$, there exists some $v$, such that $M(x)
= v$ and $v \in \tau$.
\end{defn}

\begin{lem}[Termination for expressions]
For an expression $e$, given the termination judgment $\term{\Gamma}{e}$, then
for any program state $M \in \shape(\Gamma)$, evaluating $\denote{e}\,M$ results
in some value $v$.
\end{lem}

\begin{defn}[Termination for commands]
For a command $c$, the judgment $\term{\Gamma}{c}$ is defined as: given any
program state $M \in \shape(\Gamma)$, evaluating $\denote{c}\,M$ results in a
proper distribution.
\end{defn}

In order to prove the soundness of termination rules for commands, we need a
lemma of the ``preservation'' property for well-shaped commands.

\begin{lem}[Preservation]
\label{lem:preservation}
If $M \in \shape(\Gamma)$ and $c$ is well-shaped according to $\Gamma$, then
for any $M' \in \mathtt{supp}(\denote{c}\, M)$, the program state
$M' \in \shape(\Gamma)$.
\end{lem}

\Cref{lem:preservation} tells us that well-shaped commands do not change the
data type of variables. In the sequence case, we know that both commands $c_1$
and $c_2$ terminates if executed under program states in $\shape(\Gamma)$. With
the preservation property, we know all variables in the resulting program state
after $c_1$ still have the same data type, so executing $c_2$ in the resulting
program state still terminates.

For extensions, if the extension takes a command as an argument, then the
termination of the expanded program certainly relies on the termination of the
argument. In the case of bag map and vector map, the while loop executes the
supplied map body $c$ for a finite number of iterations---the length of the
input bag or vector. Array accesses within the while loops are safe because the
index is bounded by the array length. So, requiring $c$ to terminate makes the
entire loop terminate.

For partition, the first while loop terminates because it's bounded by the value
of the supplied $\mathit{nParts}$ expression. The following bag map terminates
due to the same arguments as above. The next while loop also terminates, because
the loop is bounded by the number of entries in $\mathit{out_{idx}}$, and the
index variable $i$ is in range; the modification to $t_\mathit{part}$ also
uses an index expression that's in range.

The bag sum extension terminates since its loop is bounded by the input bag's
size, and the index variable $i$ is in range when accessing $\mathit{in}$.

\subsection{Justifying $\linear$ rules}
We will discuss the first 4 linear rules here and delay the linear rules for
extensions to \Cref{ap:bag-map-proof} because the linear property of extensions
are intimately tied to the proof of their sensitivity properties.

Recall \Cref{defn:linear-commands}: a deterministic and terminating command $c$
is linear with respect to $\Gamma_1$ and $\Gamma_2$, if for any $k \geq 0$, the
scaled typing judgment $\tytriple{k\Gamma_1}{c}{k\Gamma_2, (0, 0)}$ is true.

The rule for $\cmdskip$ is true because for any fixed $k$, we can show
$\tytriple{k\Gamma}{\cmdskip}{k\Gamma, (0, 0)}$ by reusing the core typing rule
for $\cmdskip$.

For assignment, we first prove a lemma that shows the typing relation on
expressions is linear.
\begin{lem}
Given $\hastype{\Gamma}{e}{\sigma}{\tau}$ and any $k > 0$, the judgment
$\hastype{k\Gamma}{e}{k\sigma}{\tau}$ is also true.
\end{lem}
%
Now, if we scale the pre-condition $\Gamma$ by $k$, we know $e$ has sensitivity
$k\sigma$ under $k\Gamma$. Applying the core assignment rule concludes the
proof.

For sequence of commands, since the two commands to be sequenced together are
linear with $\Gamma_1$, $\Gamma_2$ and $\Gamma_2$, $\Gamma_3$ respectively. We
know that for any $k > 0$ we pick, the typing judgments
$\tytriple{k\Gamma_1}{c}{k\Gamma_2, (0, 0)}$ and
$\tytriple{k\Gamma_2}{c}{k\Gamma_3, (0, 0)}$ hold. Now, applying the core
sequence rule concludes the proof.

In the case of conditional commands, we need another lemma that allows us to
change the post-condition in a linear typing judgment.
\begin{lem}
\label{lem:linear-conseq}
Given $\tytriple{\Gamma_1}{c}{\Gamma_2, (0, 0)}\,\linear$ and
$\Gamma_2 \leq \Gamma_3$, then $\tytriple{\Gamma_1}{c}{\Gamma_3, (0,
0)}\, \linear$ also holds.
\end{lem}
\begin{proof}
We need to show that for any $k > 0$, the typing judgement
$\tytriple{k\Gamma_1}{c}{k\Gamma_3, (0, 0)}$ is true.

From the premise, we know that $\tytriple{k\Gamma_1}{c}{k\Gamma_2, (0, 0)}$ is
true. Since $\Gamma_2 \leq \Gamma_3$, scaling both by $k$ preserves the
pointwise order $k\Gamma_2 \leq k\Gamma_3$. We are only weakening the
post-condition here, so applying the \textsc{Conseq} rule from \aprhl concludes
this proof.
\end{proof}
Using \Cref{lem:linear-conseq}, we know both branches are linear with respect to
$\Gamma$ and $\mathtt{max}(\Gamma_1, \Gamma_2)$. Then, since $e$ is
$0$-sensitive under $\Gamma$, this allows us to apply the \textsc{Cond} rule
from \aprhl concludes this case.

\section{Soundness Proofs for Extensions}
\label{ap:proofs}

\subsection{Bag-Map}
\label{ap:bag-map-proof}
To prove \textsc{Bag-Map}'s typing rule is sound, we need to show the \aprhl
judgement corresponding to the conclusion is true. This \aprhl judgement relates
two instances of the bag map program
$$\cmdbmap{{in}}{{out}}{{t_{in}}}{i}{{t_{out}}}{c}
\sim
\cmdbmap{{in}}{{out}}{{t_{in}}}{i}{{t_{out}}}{c}$$
Our general strategy for proving the soundness of these extensions is to first
prove some specification $f$ in the logic $\mathcal{L}$ specifies the
computation of the expanded extension code and then separately reason about the
sensitivity of the output from the specification $f$.

Our first step is to apply the \textsc{Equiv} rule and rewrite the pair of
related programs with an extra $\cmdskip$:
$$\cmdbmap{{in}}{{out}}{{t_{in}}}{i}{{t_{out}}}{c}; \cmdskip
\sim
 \cmdskip; \cmdbmap{{in}}{{out}}{{t_{in}}}{i}{{t_{out}}}{c}$$
\noindent And apply \textsc{Seq} followed by \textsc{While-L} and \textsc{While-R}
rule to perform one-sided reasoning. Since the two one-sided cases are
symmetric, we only discuss the case for
$\cmdbmap{{in}}{{out}}{{t_{in}}}{i}{{t_{out}}}{c} \sim \cmdskip$.

Since our goal is to give a specification of bag map, a natural choice would be
to model bag map with the $\mathtt{map}$ operator over lists, which applies a
function over each value in a bag and returns a new bag. However, we need to
find a function $f$ that adequately describes the semantics of $c$ as an argument to
$\mathtt{map}$.

We know, from the typing rule of \lstinline|bmap|, that $c$ is a deterministic
and terminating command. Here, we state an important lemma about deterministic
and terminating commands that will help us find such a function $f$.
\begin{lem}[Semantics of Deterministic Terminating Programs]
\label{lem:determ-command}
Given $\determ\, c$ and $\term{\Gamma}{c}$, then there exists a total function
$f_{\denote{c}} : \mathbb{M} \rightarrow \mathbb{M}$ such that $\denote{c}M
= \mathtt{ret}\, (f_{\denote{c}}\, M)$ for any program state $M$ in
$\shape(\Gamma)$.
\end{lem}
%
From \Cref{lem:determ-command}, we know that the semantics of $c$
can be described by a total function $f_{\denote{c}}$ mapping program states to
program states. However, the dependency analysis we performed reveals more about
$f_{\denote{c}}$. First, consider an arbitrary iteration in the bag map
loop; let's divide the program state right before executing $c$ into $4$
parts: 1)
the value of $t_{\mathit{in}}$, 2) the values of all modified variables in $c$,
3) the values of all other non-sensitive variables, and 4) the values of all
other sensitive variables. From the dependency analysis, we know $c$'s modified
variables have no dependency on 2) and 4). The variables that hold values from
3) are not modified, so their values remain constant throughout the entire loop,
and are also the same in both executions. Using this information, we can build a
function $f_{\denote{c}_\mathit{spec}}$ that takes the value of $t_\mathit{in}$
as the sole input, but calculates the same program state as $f_\denote{c}$: let $v$
be the input to $f_{\denote{c}_\mathit{spec}}$, we first create a fictitious
program state $M'$ by instantiating variables of 1), 2), and 4) in $M'$ with
well-shaped default values, and copying 3) from $M$ into $M'$. Now, by feeding
$M'[t_\mathit{in} \mapsto v]$ to $f_{\denote{c}}$, and accessing its value at
$t_{\mathit{out}}$, we get what $c$ would have computed for $t_\mathit{out}$.

Knowing these properties of $f_{\denote{c}_\mathit{spec}}$, we can choose the
one-sided invariant as:
$$\mathit{out}[0\dots i] = \mathtt{map}\, f_{\denote{c}_{\mathit{spec}}}\, \mathit{in}[0\dots i]$$
where $f_{\denote{c}_\mathit{spec}}\, v= f_\mathit{\denote{c}} M'[t_\mathit{in} \mapsto v]$.

The notation $\mathit{v}[i\dots j]$ selects a sub-array from $v$ in the range
$[i, j)$; if $j \leq i$, the it selects an empty sub-array. Thus, at the end of
both execution of bag map with $c$, we know the output bags are computed by
mapping the semantic function of $c$ over the values of the input array.

Having established this specification of bag map, we now need to consider its
sensitivity properties. We can prove the following lemma for $\mathtt{map}$:
\begin{lem}
For any function $f : \tau \rightarrow \sigma$, and two input lists $x_1$ and
$x_2$, if the bag distance between $x_1$ and $x_2$ is $d$, then the values of
$\mathtt{map}\, f\, x_1$ and $\mathtt{map}\, f\, x_2$ have bag distance up to
$d$.
\end{lem}
\begin{proof}
By induction on $d$.

When $d = 0$, the two bags must be permutations of each other. So the mapped
values are also permutations of each other. Thus, the mapped values also have
bag distance $0$.

In the inductive case, if the two input bags have distance $d+1$, without loss
of generality, assume $x_1$ has an element that is missing from
$x_2$. Then, it must be the case that $f$ applied to this element
is also missing from $\mathtt{map}\, f\, x_2$. From the induction hypothesis, we
know the mapped values of $x_1$ without this extra element and that of $x_2$
have bag distance up to $d$. Thus, adding an extra element will increase the bag
distance up to $d+1$.
\end{proof}

We also need to show that the bag map extension is $\linear$ with respect to the
pre-condition and post-condition in the conclusion of its typing rule. First, we
note that our sensitivity analysis of the bag map specification reveals the
input and output bags have the same sensitivity. So, if input bag's sensitivity
is scaled by $k$, then the output bag's sensitivity will be scaled by the same
$k$ according to this typing rule. Next, we note that all other modified
variables have sensitivity $\infty$, for any $k>0$, multiplying $k$ with
$\infty$ still results in $\infty$. For the variables not modified by bag map,
their sensitivities are not changed in this typing rule, using
the \textsc{Frame} rule from \aprhl we can show their sensitivities are also
scaled by $k$. This accounts for the linear scaling of all the variables in the
pre- and post-conditions.

\subsection{Vector-Map}
We follow the same strategy applied in the proof of bag map and first build the
function $f_{\denote{c}_\mathit{spec}}$ that characterizes $c$ behavior using
just $t_\mathit{in}$ and the values of non-modified non-sensitive variables
before entering the vector map loop. We also establish the same loop invariant:
$$\mathit{out}[0\dots i] = \mathtt{map}\, f_{\denote{c}_\mathit{spec}}\, \mathit{in}[0\dots i]$$

From the premises of the vector map typing rule, we know $c$ is a linear
command. So, by \Cref{defn:linear-commands}, we know that if $v\ag{1}$ and
$v\ag{2}$ have distance $d$, then $f_{\denote{c}_\mathit{spec}}\,v\ag{1}$
and $f_{\denote{c}_\mathit{spec}}\,v\ag{2}$ have distance $sd$, where $s
= \Gamma_3(t_\mathit{out})$, and $\Gamma_3$ is the typing context as specified
in the typing rule for vector map. In fact, we give the following definition to
characterize functions like $f_{\denote{c}_\mathit{spec}}$:

\begin{defn}
\label{defn:lipschitz}
Given a function $f : \tau \rightarrow \sigma$, if there is a number
$s \in \mathbb{R}^{> 0} \cup \{\infty\}$ such that for any $x_1, x_2\in \tau$,
the distance $d_\sigma(f\, x_1, f\, x_2) \leq sd_\tau(x_1, x_2)$, then we call
$f$ a linear function with scale factor $s$ with the chosen distance functions
$d_\tau$ and $d_\sigma$.
\end{defn}

We remark that functions as defined by \Cref{defn:lipschitz} are also called
lipschitz functions with lipschitz constant $s$. However, in the usual
mathematical definition of $s$-lipschitz functions, the value of $s$ does not
include $\infty$.

We then prove the following lemma for such functions:

\begin{lem}
\label{lem:array-scale}
If $f : \tau \rightarrow \sigma$ is a linear function with scale factor $s$,
then $\mathtt{map}\, f$ is also a linear function with scale factor $s$ with
distance functions chosen as the array distances on $[\tau]$ and $[\sigma]$.
\end{lem}
\begin{proof}
We first case analyze on whether the two input arrays have the same length.

If they have different lengths, then their distance is $\infty$, so for any
positive $s$, the scaled distance $s \cdot \infty = \infty$ is still
infinite. And we are done.

If they have the same length, then we proceed by induction on the length of
these arrays.

When the length is $0$, the two arrays have distance $0$, and so do the mapped
arrays. The inequality $0 \leq s \cdot 0 = 0$ holds, so we are done.

In the inductive case, assume both input arrays $x_1$ and $x_2$ have length
$n+1$. Consider the prefix sub-arrays of length $n$. Let $d_{\mathit{prefix}}$
be the distance between the two prefix sub-arrays, and let $d_{\mathit{last}} =
d_{[\tau]}(x_1, x_2) - d_\mathit{prefix}$. Let the distance between the mapped
prefix sub-arrays be $d_\mathit{prefix}'$.

From the induction hypothesis, we know $d_\mathit{prefix}' \leq
sd_\mathit{prefix}$. Now, let $d_\mathit{last}' = d_\sigma(f\, x_1[n], f\,
x_2[n])$. Since $f$ is a linear function with scale factor $s$, we know
$d_\mathit{last}' \leq sd_\mathit{last}$.

Combining this inequality with the previous one, we get $d_\mathit{prefix}' +
d_\mathit{last}' \leq s (d_\mathit{prefix} + d_\mathit{last})$. By the
definition of array distances, we know the distance between the mapped arrays is
less than the distance between the input arrays scaled with $s$.
\end{proof}

Applying \Cref{lem:array-scale} together with the loop invariant we established
concludes the proof for array map.

The vector map program is also linear with respect to the pre- and
post-conditions produced by its sensitivity typing rule. Since the output array
has a sensitivity in the post-condition that is a multiple of the input array's
sensitivity in the pre-condition, scaling the input array's sensitivity by some
$k>0$ will scale the output array's sensitivity by the same factor. The other
modified and non-modified variables follow the same argument from bag map.

\subsection{Partition}
\label{ap:partition-proof}
We use the same reasoning from bag map to show that $out_\mathit{idx}
= \mathtt{map}\, f_{\denote{c}_\mathit{spec}}\, \mathit{in}$ --- that is, the
partition indices are computed by mapping the semantics $f_{\denote{c}_\mathit{spec}}$ over the input
bag.

For the second while loop that places each input bag element into the
corresponding partition, we will specify it using the combinator
$\mathtt{foldl}$ over bags:
\begin{lstlisting}
foldl : (b -> a -> b) -> b -> {a} -> b
\end{lstlisting}
The first argument to $\mathtt{foldl}$ will be a function that specifies how to
place one element from the input bag into the output partitions. We will use the
function $\mathit{place}$ to build this argument:
\begin{lstlisting}
place : [{&$\tau$&}] -> (int * &$\tau$&) -> [{&$\tau$&}]
place parts (idx, elmt) =
  match nth idx parts with
    Some part -> update (part ++ [elmt]) idx parts
    None      -> parts
\end{lstlisting}
The $\mathit{place}$ function takes the partitions, followed by a pair of
partition index and the bag element, and produces new partitions such that for
pairs whose indices are in range, the bag element will be inserted at the end of
the indexed partition. The pairs whose indices are out of range are simply
ignored.

The specification of the second while loop is then
$$\mathit{out}[0\dots i] = \mathtt{foldl}\; \mathtt{place}\; \mathit{empty}\;
(\mathtt{zip}\; \mathit{out_{idx}}[0\dots i]\; \mathit{in}[0\dots i])$$
The right hand side can be expanded into
$$\mathtt{foldl}\; \mathtt{place}\; \mathit{empty}\;
(\mathtt{zip}\; (\mathtt{map}\, f_{\denote{c}_\mathit{spec}}\, \mathit{in}[0\dots
i])\; \mathit{in}[0\dots i])$$
using the specification established the first while loop.
The $\mathit{empty}$ value is the initial empty partition---an array of empty
bag values, whose length is equal to the value of the specified
$\mathit{nParts}$ parameter. The $\mathtt{zip}$ operator takes two lists and
produces a list of pairs.

We prove this specification characterizes the behavior of the second loop using
the one-sided \aprhl rules.

Next, we need to consider the sensitivity of $\mathit{out}$ using the
$\mathtt{foldl}$ specification. Our typing rule claims that the distance between
$\mathit{out}\ag{1}$ and $\mathit{out}\ag{2}$ is at most the distance between
the input bags $\mathit{in}\ag{1}$ and $\mathit{in}\ag{2}$. Let $d$ be the bag
distance between $\mathit{in}\ag{1}$ and $\mathit{in}\ag{2}$.

But first, we will need to establish a few more properties about values
in \Fuzzi.

\begin{defn}[Equivalence]
Two values $v_1$ and $v_2$ of type $\tau$ are equivalent if their distance
$d_\tau(v_1, v_2)$ is $0$.
\end{defn}

\begin{lem}
The equivalence definition is a proper equivalence relation: it is symmetric,
reflexive and transitive.
\end{lem}
\begin{proof}
By induction on the type.
\end{proof}

\begin{lem}
\label{lem:dist-equiv}
Given two values $x_1$ and $x_1'$ of type $\tau$ that are equivalent, for any
$x_2$ of the same type, $d_\tau(x_1, x_2) = d_\tau(x_1', x_2)$.
\end{lem}
\begin{proof}
By induction on the type.
\end{proof}

For bags, two bags are equivalent if they are permutations of each other.

For arrays, two arrays are equivalent if they have the same length and each
pair of elements in their corresponding positions are equivalent.

So, two arrays of bags are equivalent if these two arrays have the same length
and each pair of bags at every position are permutations of each other.

Let's establish a notation for writing down permutations. We will use a sequence
of integers to express a permutation. For example, given the bag $\{a, b, c\}$,
an equivalent bag under the permutation $312$ is $\{c,a,b\}$.

Next, we show that the partition specification respects the equivalence
relation.

\begin{lem}
\label{lem:fold-equiv}
Given two bags $x_1$ and $x_2$ and any total function $f$, if $x_1$ and $x_2$
are equivalent, then
$$\mathtt{foldl}\; \mathtt{place}\; \mathit{empty}\;
(\mathtt{zip}\; (\mathtt{map}\, f\, x_1)\; x_1)$$
and
$$\mathtt{foldl}\; \mathtt{place}\; \mathit{empty}\;
(\mathtt{zip}\; (\mathtt{map}\, f\, x_2)\; x_2)$$
are also equivalent.
\end{lem}
\begin{proof}
Let $\sigma$ be the permutation such that $\sigma(x_1) = x_2$.

We call $i, j$ an inversion in $\sigma$ if $i$ appears before $j$, but $i >
j$. We proceed by induction on the number of inversions in $\sigma$.

In the base case, $\sigma$ is identity, so $x_1 = x_2$, and we are done.

In the inductive case, let there be $k+1$ inversions in $\sigma$. There must be
an adjacent inversion in $\sigma$. An adjacent inversion is a length-$2$
subsequence $ij$ in $\sigma$ such that $i$ appears immediately before $j$, but
$i > j$. If there were no such adjacent inversions, then $\sigma$ must be
identity again, and this would be a contradiction with the number of inversions
$k+1$.

We produce a new permutation $\sigma'$ by swapping $i$ and $j$. The new
permutation $\sigma'$ must have $k$ inversions. Consider the following
illustration:
$$\sigma' = \sigma_1\sigma_2\dots j i \dots\sigma_n$$ For all numbers $\sigma_1$ upto
$\sigma_{k}$ right before $i$, their relative position did not change with
respect to $i$ or $j$, so we did not introduce nor eliminate inversions by
swapping $i$ and $j$. The same argument goes for all numbers after $j$.

Since $ij$ itself is an inversion, the total number of inversions must have
decreased by $1$.

So, by induction hypothesis, the output from folding $x_1$ and $\sigma'(x_1)$
must be equivalent.

Now, we just need to show the output from folding $\sigma'(x_1)$ and $x_2$ must
be equivalent as well, then we are done. Recall $x_2 = \sigma(x_1)$.

The only difference between these two bags is that $x_1[i]$ and $x_1[j]$ appears
in swapped orders in $\sigma'(x_1)$ and $\sigma(x_1)$.

Let's do a case analysis on whether $f\, x_1[i] = f\, x_1[j]$. Let the
partition index computed by applying $f$ to $x_1[i]$ and $x_1[j]$ be $m$ and
$n$.

If $m\neq n$, then the output from $\mathtt{foldl}$ will in fact be
identical. This is because since $x_1[i]$ is placed into the $m$-th output bag,
and $x_1[j]$ is placed into the $n$-th output bag, but the order in which they
are processed have not changed with respect to other elements in their
respective output bags. So the output remains identical.

If $m = n$, then only the $m$-th output bag will be impacted. And in that output
bag, the values $x_1[i]$ and $x_1[j]$ will be swapped since the order in which
they are processed is changed. But this just permutes the $m$-th output bag. So,
the output arrays from folding $\sigma'(x_1)$ and $x_2$ remain equivalent.
\end{proof}

By \Cref{lem:dist-equiv} and \Cref{lem:fold-equiv}, we know that for any two
inputs $x_1$ and $x_2$, permuting them does not change the distance on their
output from the specification of partition.

Knowing this, we are ready to prove the sensitivity condition for partition.

\begin{lem}
Given two bags $x_1$ and $x_2$ with distance $d$ and a total function $f$, the
distance between
$$\mathtt{foldl}\; \mathit{place}\; \mathit{empty}\;
(\mathtt{zip}\; (\mathtt{map}\, f\, x_1)\; x_1)$$
and
$$\mathtt{foldl}\; \mathit{place}\; \mathit{empty}\; (\mathtt{zip}\;
(\mathtt{map}\, f\, x_2)\; x_2)$$ is also $d$.
\end{lem}
\begin{proof}
Since we know permuting $x_1$ and $x_2$ does not change the distance on the
outputs, with out loss of generality, we can assume that $x_1$ and $x_2$ are
arranged such that $x_1 = px_1'$ and $x_2 = px_2'$, where $p$ is a common
prefix, and $x_1'$ is the multiset difference between $x_1$ and $x_2$, and
similarly $x_2'$ is the multiset difference between $x_2$ and $x_1$. This
implies $d = |x_1'| + |x_2'|$.

Proceed by induction on $d$.

In the base case, we know $x_1 = x_2 = p$, and we are done.

In the inductive case, assume the bag distance is $d+1$, and $x_1 = px_1'v$ and
$x_2 = px_2'$, where $v$ is the last element in the multiset difference between
$x_1$ and $x_2$.

Since $d + 1 = |x_1'v| + |x_2'|$, which implies $d = |x_1'| + |x_2'|$. By
induction hypothesis, we know the distance between folding $px_1'$ and $x_2$ is
$d$. We just need to consider how adding the last value $v$ to $x_1$ changes the
distance.

Let $i = f\; v$. So $v$ will be placed at the end of the $i$-th output
bag. Since we are only adding $v$ to the output from $px_1'$ and nothing is
added to the output from $x_2$, this implies the distance must increase by
$1$. This concludes the proof.
\end{proof}

\subsection{Bag-Sum}
\label{ap:bag-sum-proof}
We again deploy the same strategy used so far. We can establish a specification
that models bag sum with the one-sided loop invariant
$$\mathit{out} =
\mathtt{foldl}\, (\lambda\,\mathit{sum}\,\mathit{v}.\, \mathit{sum} + \mathtt{clip}\, \mathit{bound}\, v)\, 0\, \mathit{in}[0\dots i]$$
where $\mathtt{clip}$ returns a value whose magnitude is within the bound set by
its first argument.

We also show that this specification respects equivalence relations on bags.

\begin{lem}
\label{lem:fold-plus-equiv}
Given two equivalent bags $x_1$ and $x_2$, the values $\mathtt{foldl}\,
(\lambda\,\mathit{sum}\,\mathit{v}.\, \mathit{sum}
+ \mathtt{clip}\, \mathit{bound}\, v)\, 0\,x_1$ and $\mathtt{foldl}\,
(\lambda\,\mathit{sum}\,\mathit{v}.\, \mathit{sum}
+ \mathtt{clip}\, \mathit{bound}\, v)\, 0\,x_2$ are the same.
\end{lem}
\begin{proof}
The proof again proceeds by induction on the number of inversions in the
permutation $\sigma$ where $x_2 = \sigma(x_1)$.

In the inductive case, we apply commutativity of addition to conclude the proof.
\end{proof}

So, by \Cref{lem:dist-equiv} and \Cref{lem:fold-plus-equiv}, we can again
permute the inputs without changing the distance on the outputs.

Finally, to reason about the sensitivity of the specification, we consider the
following lemma.

\begin{lem}
Given two bags $x_1$ and $x_2$ of distance $d$ and a non-negative real number
$\mathit{bound}$, then distance between $\mathtt{foldl}\,
(\lambda\,\mathit{sum}\,\mathit{v}.\, \mathit{sum}
+ \mathtt{clip}\, \mathit{bound}\, v)\, 0\,x_1$ and $\mathtt{foldl}\,
(\lambda\,\mathit{sum}\,\mathit{v}.\, \mathit{sum}
+ \mathtt{clip}\, \mathit{bound}\, v)\, 0\,x_2$ is at most
$d\cdot \mathit{bound}$.
\end{lem}
\begin{proof}
With out loss of generality, we can again assume $x_1 = p x_1'$ and $x_2 = p
x_2'$ just like we did for partition. We know $d = |x_1'| + |x_2'|$.

Proceed by induction on $d$. The base case follows directly from $x_1 = x_2 =
p$.

In the inductive case, we assume $x_1 = px_1'v$, where $v$ is the last element
of $x_1$. Since $d = |x_1'| + |x_2'|$, by induction hypothesis, we know the
distance between bag summing $px_1'$ and $x_2$ is at most
$d\cdot \mathit{bound}$.

The bag sum of $x_1$ simply adds $\mathtt{clip}\, \mathit{bound}\, v$ to the bag
sum of $px_1'$. But the absolute value of $\mathtt{clip}\, \mathit{bound}\, v$
is at most $\mathit{bound}$. So, the distance between the bag sum of $x_1$ and
$x_2$ is at most $(d+1)\cdot \mathit{bound}$. This concludes the proof.
\end{proof}

The bag sum program is also linear with respect to the typing contexts admitted
by its sensitivity rules. Scaling the pre-condition by $k>0$ will
correspondingly cause the output sum's sensitivity to be scaled by $k$ in the
post-condition. The other modified variables have $\infty$-sensitivity, and the
scaled post-condition also has their sensitivity as $\infty$.

\subsection{Advanced Composition}
The advanced composition rule is a straightforward application of the \aprhl
advanced composition rule.

\fi 

\section{\Fuzzi Implementation of Differentially Private Gradient Descent}
\label{ap:gradient-descent}
We show the full implementation of differentially private gradient descent for
logistic regression (as discussed in \Cref{sec:eval-logistic-regression}) here.
This code shown here is largely comprised of three parts:
\begin{enumerate*}
\item a \lstinline|bmap| application that preprocesses the input data,
\item a second \lstinline|bmap| application that computes the private gradients,
\item and a final step that releases noised gradients and updates model parameter.
\end{enumerate*}
The code also uses a special extension called \lstinline|repeat|. This extension
takes a loop index variable, a constant literal integer and a \Fuzzi command as
parameters, and expands to a while loop that executes the command for the
specified number of times. The typing rule for this extension simply unrolls the
loop for the specified number of times, but perform no special deduction on the
sensitivities and privacy cost for the entire loop. We had elided this extension
from the main body of the paper because it only provides a better programming
experience (one could simply copy the loop body for the specified number of
times to reach same result), but does not provide additional insight to \Fuzzi's
design.

\newpage
\begin{lstlisting}
lamb = 0.1;
rate = 0.1;
epoch = 0;
size $= lap(10.0, fc(length(db)));
/* used advanced composition for 100 total passes */
ac(epoch, 100, 1.0e-6,
  /* extend each row to account for bias */
  bmap(db, db1, trow, i, trow1,
    trow1 = zero_786;
    trow1[0] = 1.0;
    repeat(j, 785, trow1[j+1] = trow[j];);
    j = 0;
  );
  /* compute the gradient for each row */
  i = 0;
  trow1 = zero_786;
  bmap(db1, dws, trow1, i, twout,
    twout = zero_785;
    repeat(j, 785, twout[j] = trow1[j];);
    j = 0;
    dt = clip(dot(twout, w), 100.0);
    temp = exp(-1.0 * trow1[785] * dt);
    prob = 1.0 / (1.0 + temp);
    sc = (1.0 - prob) * trow1[785];
    twout = scale(sc, twout);
    dt = 0.0;
    temp = 0.0;
    prob = 0.0;
    sc = 0.0;
  );
  /* compute noised gradient and update model parameter */
  repeat(j, 785,
    i = 0; twout = zero_785; tf_out = 0.0;
    bmap(dws, dws_j, twout, i, tf_out, tf_out = twout[j];);
    i = 0;tf_out = 0.0;
    bsum(dws_j, j_sum, i, tf_out, 1.0);
    j_sum $= lap(5000.0, j_sum);
    w[j] = w[j] + (j_sum / size - 2.0 * lamb * w[j]) * rate;
  );
  /* clear aux variables */
  db1 = {}; dt = 0.0;
  dws = {}; dws_j = {};
  i = 0; j = 0;
  prob = 0.0; sc = 0.0; temp = 0.0; tf_out = 0.0;
  trow = zero_785; trow1 = zero_786; twout = zero_785;
);
\end{lstlisting}

\end{document}